\documentclass[11pt]{article}
\usepackage{geometry}
\usepackage{authblk}
\usepackage{latexsym,amsmath,amsfonts,amssymb,amsbsy,amsopn,amstext,amsthm}
\usepackage{epsfig,epstopdf,graphicx}
\usepackage[mathscr]{eucal}
\usepackage{color,multicol,makeidx}
\usepackage{bm,mathrsfs}
\usepackage{float}
\usepackage{mathrsfs}
\usepackage{tabu}
\usepackage{enumerate}
\usepackage{indentfirst} 
\usepackage{verbatim}
\usepackage{algorithm}
\usepackage{algorithmicx}
\usepackage{algpseudocode} 
\usepackage[labelformat=simple]{subcaption}

\def\eq{\displaystyle\stackrel\triangle=}
\def\xra{\xrightarrow}
\def\ra{\rightarrow}

\newtheorem{assum}{Assumption}
\newtheorem{thm}{Theorem}
\newtheorem{prop}{Proposition}
\newtheorem{lem}{Lemma}

\newtheorem{cor}{Corollary}

\newtheorem{defn}{Definition}

\providecommand{\keywords}[1]
{
	\small	
	\textbf{\textit{Keywords---}} #1
}
\begin{document}
\title{\textbf{Theoretical Guarantees for AOA-based Localization: Consistency and Asymptotic Efficiency}}
	\author[1]{Shenghua Hu}
\author[2]{Guangyang Zeng}
\author[1]{Wenchao Xue}
\author[1]{Haitao Fang}
\author[1]{Biqiang Mu}

\affil[1]{State Key Laboratory of Mathematical Sciences, Academy of Mathematics and Systems Science, Chinese Academy of Sciences, Beijing 100190, China\\
\texttt{hushenghua17@mails.ucas.ac.cn; wenchaoxue@amss.ac.cn; htfang@iss.ac.cn; bqmu@amss.ac.cn}}
\affil[2]{School of Data Science, Chinese University of Hong Kong, Shenzhen, Shenzhen, China\\
\texttt{zengguangyang@cuhk.edu.cn}}

	\maketitle
	\begin{abstract}
We study the problem of signal source localization using angle of arrival (AOA) measurements. We begin by presenting  verifiable geometric conditions for sensor deployment that ensure the model's asymptotic localizability. Then we establish the consistency and asymptotic efficiency of the maximum likelihood (ML) estimator. However, obtaining the ML estimator is challenging due to its association with a non-convex optimization problem.
To address this, we propose an asymptotically efficient two-step estimator that matches the ML estimator’s asymptotic properties while achieving low computational complexity (linear in the number of measurements). The primary challenge lies in obtaining a consistent estimator in the first step. To achieve this, we construct a linear least-squares problem through algebraic operations on the measurement nonlinear model to first obtain a biased closed-form solution. We then eliminate the bias using the data to yield an asymptotically unbiased and consistent estimator.
In the second step, we perform a single Gauss-Newton iteration using the preliminary consistent estimator as the initial value, achieving the same asymptotic properties as the ML estimator. Finally, simulation results demonstrate the superior performance of the proposed two-step estimator for large sample sizes.
	\end{abstract}


\section{Introduction}
Signal source localization estimates transmitter positions using sensor measurements and plays a vital role in a wide range of applications, including radar, sonar,  wireless networks, cognitive radio networks, and multimedia systems \cite{Niu2018,Li2002,Gezici2005,Wang2014, Hua2015}. Common localization methods utilize time of arrival (TOA) \cite{Chan1994,Soares2015,Zeng2022}, time difference of arrival (TDOA) \cite{Ho2012,Sun2019,Zeng2024}, received signal strength (RSS) \cite{So2011,Angjelichinoski2015,Liu2016}, and  AOA \cite{Ho2006, Dogancay2006,Wang2012, Shao2014,Wang2015,Zheng2019, Sun2020, Wang2023}. 
	Among these, AOA-based techniques determine source location using angle measurements and offer distinct advantages: they do not require sensor synchronization, which is essential for TOA and TDOA, and provide high accuracy through antenna arrays while minimizing inter-sensor communication \cite{Wang2015,Wang2018,Wang2023}. These advantages make AOA-based localization methods particularly suitable for use with commercial Wi-Fi systems,  low-altitude wireless networks, and multi-drone meshes where low latency, decentralized operation, and scalability are critical \cite{Wu2019, Zheng2019, Meles2021, GomezTornero2018}. Therefore, this paper focuses on signal source localization using AOA  measurements.
	
	The AOA-based localization identifies the source by intersecting rays derived from angle measurements. 
	In practical scenarios, measurements are usually corrupted by noise, and the source coordinates must be deduced from noisy data.
	Therefore, AOA-based source localization is a typical parameter estimation problem. The most commonly utilized methods for addressing this problem are maximum likelihood (ML) and least squares (LS) methods, which are equivalent when the measurement noises are independent and identically distributed (i.i.d.) Gaussian random variables. The non-convex nature of the AOA-based localization problem poses a challenge in searching for a global minimum. 
	
	To tackle this problem, two primary strategies emerge: iteration and relaxation-based methods.
	The iteration-based methods include: Wang and Ho \cite{Wang2018} represent the source location in modified polar coordinates (MPR) and apply the Gauss-Newton (GN) algorithm to search for the ML estimator with a priori knowledge that the source is in the near field. Wang \textit{et al.} \cite{Wang2018a} apply a similar technique in the two-dimensional (2-D) scenario. Wang \textit{et al.} \cite{Wang2012} develop the location-penalized likelihood function in the 2-D scenario and apply the Broyden-Fletcher-Goldfarb-Shanno algorithm to maximize this penalized likelihood. While both Do{\u{g}}an{\c{c}}ay \cite{Dogancay2005} and Wang \textit{et al.} \cite{Wang2007} address a constrained total least squares  problem iteratively, they employ distinct problem formulations and iterative schemes.
	For relaxation-based methods: Wang and Ho \cite{Wang2015} derive pseudolinear models via Taylor-series expansion under small noise assumption, then expand the parameter space and solve a weighted least squares problem with a quadratic constraint for bias reducing. Sun \textit{et al.} \cite{Sun2020} apply similar approximation and methodology but with the problem formulated in MPR. In \cite{Wang2018, Wang2018a}, similar methodology of approximation, linearization, and formulation of a constrained weighted least squares (CWLS) problem is employed. The CWLS problem is solved using a semi-definite relaxation (SDR) technique, and the solution serves as an initial value for subsequent iterative methods. Chen \textit{et al.} \cite{Chen2021} apply a similar linearization and SDR method, but reduce estimation bias by analyzing the second order term of the Taylor-series expansion, approximating the bias and compensating for it in the solution. 
	Unlike conventional linearization approaches, 
	Luo \textit{et al.} \cite{Luo2019} construct a linear model using all AOA and angular geometric information, but its LS estimator relies on approximating unknown quantities that must be approximated from available data.

	Existing iteration and relaxation-based approaches to AOA-based localization face  limitations. The iteration-based methods  suffer from high computational complexity and often converge to a stationary point rather than the global solution. The relaxation-based methods typically introduce asymptotic bias stemming from small-noise approximations, which may not vanish even as the number of measurements increases to infinity \cite{Wang2015, Ho2006}.  Particularly, this bias increases as noise level increases. Overall, existing iteration and relaxation-based methods lack consistency and asymptotic efficiency for AOA-based localization\footnote{``Consistency" refers to the estimator converging to the true value as the number of measurements increases, while ``asymptotic efficiency" means that, as the number of measurements grows, the mean squared error (MSE) of the estimator approaches the Cramér-Rao lower bound (CRLB), the theoretical minimum MSE for any unbiased estimator.}. 
Furthermore, the fundamental question of asymptotic localizability\footnote{``Asymptotic localizability" describes the property that the source location can be uniquely determined asymptotically using available measurements. This concept is rigorously formalized in Definition \ref{def_asy_loc}.} for AOA-based localization has yet to be sufficiently explored in the literature.
To address these gaps, this paper makes two key contributions for AOA-based localization:
  (i) we establish the geometric conditions for asymptotic localizability; (ii) we develop an estimator that is consistent and asymptotically efficient with low computational complexity.
	
	
To establish asymptotic localizability, we develop verifiable geometric conditions for sensor deployment, which ensures that the sensor network can uniquely determine the true source asymptotically.
To obtain a consistent and asymptotically efficient estimator for AOA-based localization, we first develop the ML estimator and rigorously prove its consistency and asymptotic efficiency under suitable conditions. However, the non-convex nature of the ML estimator poses challenges for accurate and efficient computation.
To address the challenges, we introduce an asymptotically efficient two-step estimator: (i) deriving a consistent estimator of the source location; (ii) running a single GN iteration using this consistent estimator  as the initial value.
A key advantage of the two-step estimator is its guaranteed asymptotic efficiency when the initial consistent estimator
has a convergence rate of $O_p(1/\sqrt{n})$ \cite{Lehmann1998}, where $n$ is the number of measurements.
The two-step estimator has been successfully applied in parameter estimation of nonlinear rational models\cite{Mu2017}, and TOA and TDOA-based localization \cite{Zeng2022,Zeng2024}.
 Therefore, the key to obtaining the asymptotically efficient estimator of the AOA-based localization lies in deriving a  consistent estimator with the convergence rate of $O_p(1/\sqrt{n})$.
 This is accomplished through the bias eliminated least squares (BELS) estimators.

	 Deriving a consistent estimator for AOA-based localization is challenging due to the composite structure of anti-trigonometric and norm functions in the measurement model, as shown in \eqref{mea_2d} and \eqref{mea_3d}.
For both 2-D and 3-D scenarios, through appropriate model transformation, we obtain a linear regression model where the parameter vector comprises the true source. The resulting BELS estimator provides a consistent estimator for the true source with the rate $O_p(1/\sqrt{n})$.
The BELS estimator requires the variance of the sine of the noise, which is directly related to the noise variance. In the case of unknown noise variance, we develop a $\sqrt{n}$-consistent estimator of the variance of the sine of the noise.
The computational complexity of our two-step estimator is dominated by the LS estimation and a single GN iteration, resulting in linear scaling with the number of measurements and ensuring high efficiency.

The proposed two-step estimator achieves asymptotically optimal accuracy with low computational complexity, making it suitable for large-scale AOA-based localization. Applications include satellite constellations—comprising up to 40,000 units equipped with low-cost spaceborne interferometers \cite{Li2022}—which can provide extensive AOA measurements. In addition, massive MIMO radar systems using millimeter-wave virtual arrays enable fast AOA estimation \cite{Godrich2010,Li2022a}, supporting accurate localization of slow-moving targets under high measurement rates.

	The rest of the paper is organized as follows. Section \ref{sec2} addresses the 2-D AOA-based localization problem, focusing on the asymptotic localizability, ML estimator, and   asymptotically efficient two-step estimator. Section \ref{sec3} extends the 2-D AOA localization to the 3-D case. Section \ref{sec4} demonstrates the estimation accuracy and computational efficiency of the estimator using extensive Monte-Carlo simulations. Finally, Section \ref{sec5} concludes the paper.
	 The proofs of all theoretical results are put in the supplemental file.
	
\textbf{Notation:} 
Throughout this paper, 
\(\mathbb{E}\), \(\mathbb{V}\), and  \(\mathbb{P}\) denote the expectation, variance, and probability, respectively, with respect to  the noise, unless otherwise specified.
  The superscript \((\cdot)^o\) indicates the true or noise-free value of a given quantity. For a sequence of random variables \(\{X_n\}\), (i)  \(X_n = O_p(1)\) means that \(X_n\) is bounded in probability (i.e.,  for any $\epsilon> 0$, there exist constant $L>0$ and integer $n_0>0$ such that $\mathbb{P}(|X_n|>L)<\epsilon$ for $n>n_0$); (ii) \(X_n = o_p(1)\) indicates that \(X_n\) converges to zero in probability (i.e., 
 for any $\epsilon> 0$,  there holds that $\lim_{n\xra{}\infty}\mathbb{P}(|X_n|>\epsilon)=0$); 
 (iii) $X_n\xra{}\mathcal{N}(0,\sigma^2)$ denotes that $X_n$ converges in distribution to a Gaussian random variable with mean zero and variance $\sigma^2$. Lastly, \(\nabla\) and \(\nabla^2\) denote the first and second order differential operators, respectively.
	
	\section{2-D scenario}
	\label{sec2}
	We first consider the 2-D scenario where the sensors and the signal source are in a 2-D space. This corresponds to the localization problem on the water surface or on the flat ground, for instance. Meanwhile, the 2-D problem is also a part of the 3-D problem that will be addressed later.
	\subsection{Problem formulation}
	Suppose there are \( n \) sensors distributed in a 2-D space, each with precisely known coordinates \( p_i = [x_i, y_i]^T \) for \( i = 1, \ldots, n \). Let \( p^o = [x^o, y^o]^T \) represent the unknown coordinates of the true source that need to be estimated using AOA  measurements from sensors.
	Fig. \ref{fig:2D} illustrates the scenario of the 2-D source localization using a AOA  sensor array, where $a_i^o$ is the noise-free AOA  measurement obtained by sensor $i$, while the AOA  measurement of the source obtained by sensor $i$ is:
	\begin{equation}
		\label{mea_2d}
		a_i = \arctan\left(\frac{y_i-y^o}{x_i-x^o} \right) + \varepsilon_i^a, \quad i=1,...,n,
	\end{equation}
	where $\varepsilon_i^a$ is the measurement noise. Our goal is to estimate $p^o$ from $\{p_i\}_{i=1}^n$ and $\{a_i\}_{i=1}^n$ according to the measurement model \eqref{mea_2d}. For the measurement noise, we make the following assumption.
	
	\begin{figure}[!t]
		\centering
		\includegraphics[width=0.4\textwidth]{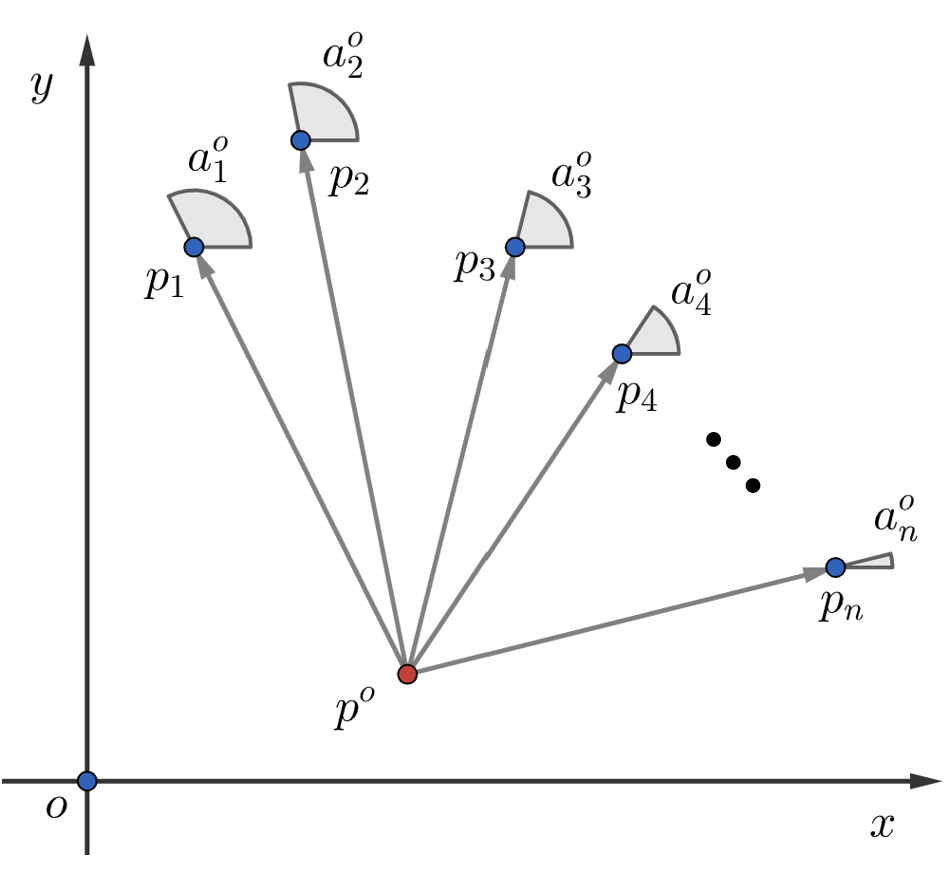}
		\caption{Illustration of AOA  measurements in the 2-D scenario. The red dot represents the signal source while the blue dots represent the sensors.}
		\label{fig:2D}
	\end{figure}

	\begin{assum}
		\label{assum_mea_noise_2d}
		The measurement noises $\{\varepsilon_i^a\}_{i=1}^n$ are i.i.d. Gaussian random variables with mean zero and finite variance $\sigma_a^2>0$.
	\end{assum}

\subsection{Asymptotic localizability}
In this section, we present sufficient conditions on sensor geometric deployment to ensure the asymptotic localizability of AOA-based localization.

	\begin{assum}
		\label{assum_coordinates_2d}
		The source $p^o$ fall in a bounded set $\mathcal{P}^o$ and the sensors $p_i$, $i=1,...,n$, belong to a bounded set $\mathcal{P}$, independent of $n$. Moreover, the sets are disjoint, i.e., $\mathcal{P}^o \cap \mathcal{P} = \emptyset$. 
	\end{assum}

In practice, Assumption \ref{assum_coordinates_2d} is readily fulfilled because sensors cannot be positioned arbitrarily far from the source. Additionally, the set \(\mathcal{P}^o\) can always be defined as a compact region surrounding the source, without needing prior knowledge of the source's precise location.

	\begin{defn}
\label{def:empirical_dist}
Let \( x_1, \ldots, x_n \) be a sequence in a measurable space \((\Omega, \mathcal{F})\). The empirical distribution \( F_n \) is the discrete probability measure on \((\Omega, \mathcal{F})\) defined by
$
F_n \eq\frac{1}{n} \sum_{i=1}^n \mathbf{1}_{\{x_i \in A\}},
$
where \(\mathbf{1}_{\{\cdot\}}\) is the indicator function.
\end{defn}
	
\begin{assum}
	\label{assum_colinearity}
	\begin{enumerate}[(i)]
	    \item The empirical distribution function $F_n$ of the sensor sequence $p_1, p_2,...$ converges to a distribution function $F_\mu$ and the probability measure induced by the distribution \(F_\mu\)  is denoted by \(\mu\);
	    \item For any positive integer \(n > 2\), the sensors \(p_1, \dots, p_n\) do not lie on a line.  
Moreover, there does not exist any subset $\mathcal{P}'$ of \(\mathcal{P}\) with \(\mu(\mathcal{P}') = 1\) such that \(\mathcal{P}'\) lies entirely on a line.
	\end{enumerate}
\end{assum}
	
	For a candidate point $p=[x,y]^T\in \mathcal{P}^o$,  
	denote the  predictive AOA of sensor $i$  by $f_i(p) \eq \arctan\left((y_i-y)/(x_i-x)\right)$, $i=1,...,n$, and define
$
h_n(p) \eq   \frac{1}{n}\sum_{i=1}^n(f_i(p)-f_i(p^o))^2
$.
We now present the definition of asymptotic localizability for the AOA-based localization problem.

\begin{defn}
\label{def_asy_loc}
For  the model \eqref{mea_2d}, the true signal source $p^o$ is called asymptotically localizable if   $h_n(p)$ has a limit function and the limit function, denoted by   $h(p)\eq \lim_{n \to \infty}h_n(p)$,  has a unique minimum at $p=p^o$.
\end{defn}

 We have the following results on the asymptotic localizability for  the AOA-based localization problem.

	\begin{thm}
		\label{thm_uniq_2d}
		Under Assumptions \ref{assum_coordinates_2d}-\ref{assum_colinearity}, the true signal source is asymptotically  localizable.
	\end{thm}
	
	\subsection{Maximum likelihood estimator}
	
This subsection  derives the ML estimator for   the AOA-based localization problem defined in \eqref{mea_2d} and presents its consistency and asymptotic efficiency.

	Based on Assumption \ref{assum_mea_noise_2d}, the log-likelihood function of the model \eqref{mea_2d} is:
	\begin{equation*}
		\ell_n(p) \eq -n \ln(\sqrt{2\pi}\sigma_a) - \frac{1}{2\sigma_a^2}\sum_{i=1}^n\left(a_i-\arctan\left(\frac{y_i-y}{x_i-x}\right)\right)^2,
	\end{equation*}
	where $\ln(\cdot)$ denotes the natural logarithm of a positive number.
The associated ML estimation problem is
	\begin{equation}
		\label{ML_2d}
 \min_{p = [x,y]^T \in \mathbb{R}^2}~\frac{1}{n}\sum_{i=1}^n\left(a_i-\arctan\left(\frac{y_i-y}{x_i-x}\right)\right)^2,
	\end{equation}
	which is equivalent to maximize $\ell_n(p)$. 
	Denote the ML estimator as $\widehat{p}_n^{\rm ML}$, which maximizes $\ell_n(p)$.
	
	Let $\nabla f_i(p)$ be the gradient vector of $f_i(p)$ over $p$, i.e.,
	\begin{align*}
		\nabla f_{i}(p) \eq \left[\frac{\partial f_i(p)}{\partial x},~\frac{\partial f_i(p)}{\partial y}\right]^T \!\!= \left[ \frac{y_i-y}{\|p_i-p\|^2},~\frac{-x_i+x}{\|p_i-p\|^2}   \right]^T\!\!.
	\end{align*}
	
	To derive the asymptotic variance of the ML estimator, we need the following lemma.
	
	\begin{lem}
		\label{lem_converge_M}
		Under Assumptions \ref{assum_coordinates_2d}-\ref{assum_colinearity}, we have: 
		\begin{enumerate}[(i)]
		    \item The matrix $ 
		 \frac{1}{n}\sum_{i=1}^n\nabla f_i(p)\nabla f_i(p)^T $ converges uniformly for $p \in \mathcal{P}^o$  as $n \to \infty$.
		  \item The limit $M^o=\lim_{n\to \infty}\frac{1}{n}\sum_{i=1}^n\nabla f_i(p^o) \nabla f_i(p^o)^T$ is nonsingular.
		\end{enumerate}
		
	\end{lem}

	Now, we are on the point to give the asymptotic property of the ML estimator. Under Assumptions \ref{assum_mea_noise_2d}-\ref{assum_colinearity}, the ML estimator  embraces the following consistency and asymptotic normality. The proof is straightforward by checking the conditions in \cite[Theorem 7]{Jennrich1969}, and we omit the proof.
	
	\begin{thm}
		\label{thm_consis_asymnormal_2d}
		Under Assumptions \ref{assum_mea_noise_2d}-\ref{assum_colinearity}, we have $\widehat{p}_n^{\rm ML} \to p^o$ almost surely as $n \to \infty$ with the asymptotic rate
		\begin{equation}
			\label{asy_normal_2d}
			\sqrt{n}(\widehat{p}_n^{\rm ML}-p^o) \to \mathcal{N}\left(0,\sigma_a^2\left(M^o\right)^{-1}\right) \quad {\rm as} ~ n\to \infty.
		\end{equation}
	\end{thm}
	
	The matrix $M^o$ is tightly related to the Fisher information matrix $F$ of model \eqref{mea_2d}. 
	To see this, firstly we have
	\begin{equation*}
		\frac{\partial \ell_n(p^o)}{\partial p^o} = \frac{1}{\sigma_a^2}\sum_{i=1}^n \varepsilon_i^a \frac{1}{\|p_i-p^o\|^2} \begin{bmatrix}
			y_i-y^o\\-x_i+x^o
		\end{bmatrix}.
	\end{equation*}
	Then we obtain the Fisher information matrix
	\begin{align*}
		F =\mathbb{E}\left[ \frac{\partial \ell_n(p^o)}{\partial p^o} \left(\frac{\partial \ell_n(p^o)}{\partial p^o}\right)^T \right]	= \frac{1}{\sigma_a^2} \sum_{i=1}^n\frac{1}{\|p_i-p^o\|^4} \begin{bmatrix}
			y_i-y^o\\-x_i+x^o
		\end{bmatrix}\left[y_i-y^o,  -x_i+x^o\right].
	\end{align*}
	This means $\lim_{n\to \infty}nF^{-1} = \sigma_a^2 \left(M^o\right)^{-1}$, which implies that the ML estimator attains the CRLB and is asymptotically efficient. 
	
	\subsection{Asymptotically efficient two-step estimator}
	\label{sec_2d_ml_twostep}
	The ML estimator owns consistency and asymptotic efficiency, but it is difficult to obtain the  global solution to the ML problem~\eqref{ML_2d} due to its non-convexity. In this subsection, we propose the asymptotically efficient two-step estimator 
	for the AOA-based localization, which has the same asymptotic property that the ML estimator possesses.  

\subsubsection{The framework of the two-step estimator}
Firstly, we prove that the  objective function of the ML problem converges to a function that is convex in a small neighborhood around $p^o$, which forms the feasibility of the two-step scheme.
	
	\begin{prop}
		\label{prop_converge_ML_objfunc}
		Under Assumptions \ref{assum_mea_noise_2d}-\ref{assum_colinearity}, $\ell_n(p)/n$ converges uniformly to 
		\begin{equation*}
			\ell(p)\eq -\ln(\sqrt{2\pi}\sigma_a)-\frac{1}{2}-\frac{1}{2\sigma_a^2}h(p)
		\end{equation*}
		on $\mathcal{P}^o$  as $n \to \infty$. In addition, $\nabla^2(-\ell(p^o)) = M^o/\sigma_a^2$ is positive definite.
	\end{prop}
	
Proposition \ref{prop_converge_ML_objfunc} indicates that $-\ell_n(p)/n$ is a convex function in a small neighborhood around the global minimum $p^o$ when $n$ is large. Therefore,  iterative methods—such as GN iterations—can be used to find the global minimum of the non-convex problem, provided that the initial value lies within this region of attraction.
Notably, a consistent estimator can approach the true value arbitrarily closely as 
$n$ increases.
Based on this observation, the two-step estimator is formally presented as follows \cite{Gourieroux1995,Mu2017}:

\textbf{Step 1.} Derive a  consistent estimator $\widehat{p}_n$ for the source's coordinates $p^o$.

\textbf{Step 2.} Run one-step GN refinement with this consistent estimator $\widehat{p}_n$ as its initial value.
This is 
\begin{subequations}\label{gn_2d}
\begin{align} 
\widehat{p}_n^{\rm GN} &= \widehat{p}_n +  \left(J_n^T J_n\right)^{-1}J_n^T(a - f),\\
J_n& = \left[\nabla f_1(\widehat{p}_n),  \cdots , \nabla f_n(\widehat{p}_n)\right]^T,\\
        f& = [f_1(\widehat{p}_n),  \cdots , f_n(\widehat{p}_n)]^T, 
    ~a =  [a_1,\cdots, a_n]^T.
\end{align}
\end{subequations}
	
The two-step scheme described above has the following appealing property, which is presented in the following lemma.
	
	\begin{lem}{\cite[Theorem 2]{Mu2017}\cite[Chapter 6, Theorem 4.3]{Lehmann1998}}
		\label{lem_onestepGN}
		Suppose that $\widehat{p}_n$ is a $\sqrt{n}$-consistent estimator of $p^o$, i.e., $\widehat{p}_n-p^o = O_p(1/\sqrt{n})$. Then under Assumptions \ref{assum_mea_noise_2d}-\ref{assum_colinearity}, we have
$	\widehat{p}_n^{\rm GN}-\widehat{p}_n^{\rm ML} = o_p(1/\sqrt{n})$.
	\end{lem}
Lemma \ref{lem_onestepGN} shows that the estimator $\widehat{p}_n^{\rm GN}$ has the same asymptotic property that $\widehat{p}_n^{\rm ML}$ possesses and hence it
is both consistent and asymptotically efficient. Thus, the success of the two-step estimation scheme hinges on obtaining a $\sqrt{n}$-consistent estimator in the first step. In the next  subsubsection, we derive such an estimator.
	
	\subsubsection{ $\sqrt{n}$-consistent initial  estimator}
	
	Denote the true angle by $a_i^o \eq \arctan\left((y_i-y^o)/(x_i-x^o)\right)$ and the true distance by $r_i^o \eq \sqrt{(x_i-x^o)^2 + (y_i-y^o)^2}$  for $i=1,...,n$.
	Thus, the model \eqref{mea_2d} is equivalent to 
	\begin{equation}
		\label{linear_2d}
		(x_i-x^o)\sin(a_i) - (y_i-y^o)\cos(a_i) = r_i^o \sin(\varepsilon_i^a).
	\end{equation}

	Define $h_i \eq [\sin(a_i),  -\cos(a_i)]^T$. We rewrite \eqref{linear_2d} in the following linear regression form \cite{Lingren1978,Aidala1979,Nardone1984}:
	\begin{equation}
		\label{linear_2d_mat}
		h_i^T p_i = h_i^T p^o + r_i^o\sin(\varepsilon_i^a), \quad i=1,\cdots,n.
	\end{equation}
	Note that the noise term  involves $\sin(\varepsilon_i^a)$, which has the following
	properties by Lemma B3 (see supplementary material):
				\begin{align}
				\mathbb{E}(\sin(\varepsilon_i^a)) = 0,
				~~	\mathbb{V}(\sin(\varepsilon_i^a)) = \frac{1}{2}(1-e^{-2\sigma_a^2}).\label{vsin}
				\end{align}
			
	By stacking \eqref{linear_2d_mat} for $n$ sensors, we obtain the following matrix form
	\begin{equation}
		Y = Xp^o + V,\label{eq8}
	\end{equation}
	where the $i$-th element of vector $Y$ is $h_i^Tp_i$, the $i$-th row of matrix $X$ is $h_i^T$, and the $i$-th element of vector $V$ is $r_i^o\sin(\varepsilon_i^a)$, for $i=1,...,n$. Then, the LS estimator corresponding to \eqref{eq8} is given by
	\begin{equation}
		\label{p_biased}
		\widehat{p}^{\rm B}_{n} = \left(X^TX\right)^{-1}X^TY.
	\end{equation}
	For the linear regression model \eqref{eq8}, note that the regressor matrix  $X$ is correlated with the noise vector $V$ due to the measurement noises $\{\varepsilon_i^a\}_{i=1}^n$.
	As a results, the LS estimator \eqref{p_biased} is biased.
	In the following, we will eliminate the bias of the LS estimator \eqref{p_biased} and eventually obtain an asymptotically unbiased and consistent estimator. 
	For achieving this, we establish a new linear regression form different from \eqref{eq8} based on the noise-free counterpart $X^o$ of $X$. Denote the regressors $h_i^o \eq [\sin(a_i^o),  -\cos(a_i^o)]^T$, which are the noise-free counterparts of $h_i$.
	Thus,   $h_i$ are connected with $h_i^o$ in the following way
	\begin{equation*}
		h_i = \cos(\varepsilon_i^a)h_i^o + \sin(\varepsilon_i^a)[\cos(a_i^o), ~ \sin(a_i^o)]^T,~ i=1,...,n,
	\end{equation*}
	and based on which  
	we rewrite  the linear regression \eqref{linear_2d_mat} as follows
	\begin{align}
		\notag
		&h_i^Tp_i =~ e^{-\sigma_a^2/2}(h_i^o)^Tp^o +w_i,~ i=1,...,n,\\
		&w_i=\big(\cos(\varepsilon_i^a)-e^{-\sigma_a^2/2}\big)(h_i^o)^Tp^o+\sin(\varepsilon_i^a)[\cos(a_i^o),  \sin(a_i^o)]p^o+\sin(\varepsilon_i^a)r_i^o. \label{linear_2d_biaseli}
	\end{align}
	The matrix form of \eqref{linear_2d_biaseli} is
	\begin{equation}
		\label{linear_2d_biaseli_mat}
		Y = X^op^o+W,
	\end{equation}
	where the $i$-th row of matrix $X^o$ is $e^{-\sigma_a^2/2}(h_i^o)^T$ and the $i$-th element of vector $W$ is $w_i$, for $i=1,...,n$.

Given that \(\{w_i\}_{i=1}^n\) is an independent random variable sequence with zero mean and finite variance, the LS estimator of the model in \eqref{linear_2d_biaseli_mat} is consistent and \(\sqrt{n}\)-consistent, provided that the Gram matrix \((X^o)^TX^o/n\) is nonsingular. This is established in the following proposition.
	
	\begin{prop}
		\label{prop_noise_free_gram}
		Under Assumptions \ref{assum_coordinates_2d}-\ref{assum_colinearity},
		the matrix $\left(X^o\right)^TX^o/n$ is nonsingular, and its limit
		exists and is nonsingular.	
	\end{prop}
	
	The non-singularity of   $\left(X^o\right)^TX^o/n$ derived in Proposition \ref{prop_noise_free_gram} provides the theoretical foundation for using the LS estimator of the linear regression model \eqref{linear_2d_biaseli_mat} to estimate $p^o$:
	\begin{equation}
		\label{unbiased_esti_2d}
		\widehat{p}_{n}^{\rm UB} = \left(\left(X^o\right)^TX^o\right)^{-1}\left(X^o\right)^TY.
	\end{equation} 

	\begin{prop}
		\label{thm_ub_2d}
		Under Assumptions \ref{assum_mea_noise_2d}-\ref{assum_colinearity}, the LS estimator $\widehat{p}_n^{\rm UB}$ is unbiased and further $\sqrt{n}$-consistent, i.e., $\widehat{p}_{n}^{\rm UB}-p^o = O_p(1/\sqrt{n})$.
	\end{prop}

	So far, we have obtained two  LS estimators, \eqref{p_biased} and \eqref{unbiased_esti_2d}, for the measurement model \eqref{mea_2d}. On one hand, the LS estimator \eqref{p_biased} is implementable using the available data but it is biased. On the other hand, the LS estimator \eqref{unbiased_esti_2d} is $\sqrt{n}$-consistent but cannot be implemented with the given data since the regressor matrix \(X^o\) is not available.
	To derive an implementable $\sqrt{n}$-consistent estimator, which is the main focus of this subsubsection, we aim to establish a direct relationship between these two LS estimators. Specifically, we intend to eliminate the bias of the LS estimator \eqref{p_biased} by calculating the non-negligible differences between \(X^TX/n\) and \((X^o)^TX^o/n\), and between \(X^TY/n\) and \((X^o)^TY/n\). We present the BELS estimator
	\begin{align}
		\widehat{p}^{\rm BE}_{n} = &~ \left(\frac{1}{n}X^TX - \mathbb{V}(\sin(\varepsilon_1^a)) I_2\right)^{-1} \left(\frac{1}{n}X^TY-\mathbb{V}(\sin(\varepsilon_1^a))\Big(\frac{1}{n}\sum_{i=1}^{n}p_i\Big)
		\right), \label{biaseli_esti_2d}
	\end{align}
	from which we clearly see the differences between the two  LS estimators.
	
	\begin{thm}
		\label{thm_be_2d}
		Under Assumptions \ref{assum_mea_noise_2d}-\ref{assum_colinearity}, the BELS estimator is $\sqrt{n}$-consistent, i.e., $\widehat{p}^{\rm BE}_{n}- p^o = O_p(1/\sqrt{n})$.
	\end{thm}

	If \(\mathbb{V}(\sin(\varepsilon_1^a))\), or equivalently \(\sigma^2_a\), is not available, then a consistent estimator for \(\mathbb{V}(\sin(\varepsilon_1^a))\) can achieve the same goal. This result is presented in the following corollary as a direct extension of Theorem~\ref{thm_be_2d}.
	\begin{cor}
		\label{cor_consis_var_2d}
		In the case that $\sigma_a^2$ is unknown, the BELS estimator $\widehat{p}_{n}^{\rm BE}$ is still $\sqrt{n}$-consistent if  $\mathbb{V}(\sin(\varepsilon_1^a))$ is replaced by a $\sqrt{n}$-consistent estimator.
	\end{cor}

	Next, we will derive a $\sqrt{n}$-consistent estimator for $\mathbb{V}(\sin(\varepsilon_1^a))$, drawing inspiration from \cite[Proposition 2]{Stoica1982}. Define
	\begin{equation}
		\label{QS}
		\!\!\!\!Q_n \eq \frac{1}{n}\!\begin{bmatrix}
			X^T	\\
			Y^T	\end{bmatrix}\!\!\big[X ~ Y\big], ~
		S_n \eq \begin{bmatrix}
			I_2	&\frac{1}{n}\sum\limits_{i=1}^{n}p_i  \\
			\frac{1}{n}\sum\limits_{i=1}^{n}p_i^T	& \frac{1}{n}\sum\limits_{i=1}^{n}\|p_i\|^2
		\end{bmatrix}.
	\end{equation}
	We proposed the estimator for $\mathbb{V}(\sin(\varepsilon_1^a))$  given by
	\begin{equation}
		\widehat{v}^a_n = \frac{1}{\lambda_{\rm max}\left(Q_n^{-1}S_n\right)},\label{sina}
	\end{equation}
	where $\lambda_{\rm max}(\cdot)$ denotes the maximum eigenvalue of a square matrix.	The following theorem shows that $\widehat{v}^a_n$ is a $\sqrt{n}$-consistent estimator of $\mathbb{V}(\sin(\varepsilon_1^a))$.
	
	\begin{thm}
		\label{thm_var_2d}
		Under Assumptions \ref{assum_mea_noise_2d}-\ref{assum_colinearity},
		$\widehat{v}^a_n$ is a $\sqrt{n}$-consistent estimator of $\mathbb{V}(\sin(\varepsilon_1^a))$, i.e.,  $\widehat{v}^a_n-\mathbb{V}(\sin(\varepsilon_1^a)) = O_p(1/\sqrt{n})$.
	\end{thm}

Algorithm \ref{algo1} displays the estimation procedure for the variance of the sine of the noise.
	
	\begin{algorithm}
		\caption{The estimation algorithm for the variance of the sine of the noise (2-D scenario)}  
		\label{algo1}
		\begin{algorithmic}[1]
			\Require Sensor locations $\{p_i\}_{i=1}^n$ and noisy AOA  measurements $\{a_i\}_{i=1}^n$.
			\State Calculate $Q_n$ and $S_n$ according to \eqref{QS};
			\State Calculate the maximum eigenvalue of $Q_n^{-1}S_n$;
			\Ensure \!The estimate for $\mathbb{V}(\sin(\varepsilon_1^a))$: $ \widehat{v}^a_n= 1/\lambda_{\rm max}(Q_n^{-1}\!S_n)$.
		\end{algorithmic}
	\end{algorithm}
	
	Algorithm \ref{algo2} presents the complete procedure for source localization using the proposed two-step estimator.
	\begin{algorithm}
		\caption{The estimation algorithm for consistent and asymptotically efficient two-step estimator   using AOA  measurements (2-D scenario)}  
		\label{algo2}
		\begin{algorithmic}[1]
			\Require Sensor locations $\{p_i\}_{i=1}^n$, AOA  measurements $\{a_i\}_{i=1}^n$, and noise variance $\sigma_a^2$ (if available).
			\If {$\sigma_a^2$ is available}
			\State Calculate the variance of the sine of noise via  \eqref{vsin};
			\State Calculate the BELS estimate $\widehat{p}_n^{\rm BE}$ according to \eqref{biaseli_esti_2d};
			\Else
			\State Run Algorithm \ref{algo1} to obtain the estimate $\widehat{v}^a_n$ of the variance of the sine of the noise;
			\State Calculate the BELS estimate $\widehat{p}_n^{\rm BE}$ according to \eqref{biaseli_esti_2d} with $\mathbb{V}(\sin(\varepsilon_1^a))$ being replaced by  $\widehat{v}^a_n$;
			\EndIf
			
			\State Run  one-step GN iteration \eqref{gn_2d} for $\widehat{p}_n^{\rm BE}$ to obtain $\widehat{p}_n^{\rm GN}$;
			\Ensure The source location estimate $\widehat{p}_n^{\rm GN}$.
		\end{algorithmic}
	\end{algorithm}
	
	In Algorithm \ref{algo1}, note that $[X~Y] \in \mathbb{R}^{n \times 3}$ and $S_n \in \mathbb{R}^{3 \times 3}$. The computational complexity is primarily determined by the matrix summation and multiplication in Line 1. Consequently, the computational complexity of Algorithm \ref{algo1} is $\mathcal{O}(n)$. For Algorithm \ref{algo2}, with $X \in \mathbb{R}^{n \times 2}$ and $Y \in \mathbb{R}^{n \times 1}$, the computational complexity of Line 3 or 6 and Line 8 is also $\mathcal{O}(n)$. Therefore, the overall computational complexity of Algorithm \ref{algo2} is $\mathcal{O}(n)$, meaning that its execution time increases linearly with the number of measurements.
	
	\section{3-D scenario}
	\label{sec3}
	In this section, we focus on the more complex 3-D scenario, an extension of the 2-D case where both the sensors and the signal source are located in a 3-D space.

	\subsection{Problem formulation}
	Suppose there are \( n \) sensors distributed in a 3-D space, each with precisely known coordinates \( p_i = [x_i, y_i, z_i]^T \) for \( i = 1, \ldots, n \). Let \( p^o = [x^o,  y^o,  z^o]^T \) represent the unknown coordinates of the source that need to be estimated using AOA  measurements from the sensors. 
	Formally, the AOA  measurement of the signal source obtained by sensor \( i =1,...,n \) is given by
	\begin{subequations}
		\label{mea_3d}
		\begin{align}
			\label{mea_3d_b}
			a_i &= \arctan\left(\frac{y_i-y^o}{x_i-x^o} \right) + \varepsilon_i^a,\\
			\label{mea_3d_e}
			e_i &= \arctan\left(\frac{z_i-z^o}{\sqrt{(x_i-x^o)^2 + (y_i-y^o)^2}} \right) + \varepsilon_i^e,
		\end{align}
	\end{subequations}
	where $\varepsilon_i^a$ and $\varepsilon_i^e$ are the measurement noises. The goal is to estimate $p^o$ from $\{p_i\}_{i=1}^n$ and $\{a_i,e_i\}_{i=1}^n$ according to the measurement models \eqref{mea_3d}. For the measurement noises, we make the following assumption.
	
		\begin{assum}
		\label{assum_mea_noise_3d}
		\begin{enumerate}[(i)]
			\item The azimuth angle measurement noises $\{\varepsilon_i^a\}_{i=1}^n$ are i.i.d. Gaussian random variables with mean zero and finite variance $\sigma_a^2>0$. 
			\item The elevation angle measurement noises $\{\varepsilon_i^e\}_{i=1}^n$ are i.i.d. Gaussian random variables with mean zero and finite variance $\sigma_e^2>0$. \item $\{\varepsilon_i^a\}_{i=1}^n$ and $\{\varepsilon_i^e\}_{i=1}^n$ are mutually independent.
		\end{enumerate}
	\end{assum}
	
\subsection{Asymptotic localizability}

	Similar to the 2-D case, in this subsection, we present sufficient conditions on sensor geometric deployment to ensure the asymptotic localizability of AOA-based localization.  Let $(p_i)_{1:2} \eq [x_i,  y_i]^T$ be the first two coordinates of $p_i$ for $i=1,...,n$. We make the assumptions on the coordinates of the signal source and sensors.
	
	\begin{assum}
		\label{assum_coordinates_3d}
		The source $p^o$ lies within a bounded set $\mathcal{P}^o$ and all the sensors $p_i,i=1,...,n$ belong to a bounded set $\mathcal{P}$, regardless of $n$. Moreover,  $\mathcal{P}^o \cap \mathcal{P} = \emptyset$.
	\end{assum}
	
	\begin{assum}
		\label{assum_projection}
		 
		\begin{enumerate}[(i)]
		
			\item $\mathcal{P}^o_{1:2} \cap \mathcal{P}_{1:2} = \emptyset$, where $\mathcal{P}^o_{1:2} \eq \{[x^o,  y^o]^T ~|~ [x^o,  y^o,  z^o]^T \in \mathcal{P}^o\}$ and $\mathcal{P}_{1:2} \eq \{[x,  y]^T ~|~ [x,  y,  z]^T \in \mathcal{P}\}$;
			\item The empirical distribution function $F_n$ of the sequence $p_1, p_2,...$ converges to a distribution function $F_\mu$;
			\item 
			For any positive integer $n$, $p_1,...,p_n$ do not lie on a line. Further, there does not exist   a subset $\mathcal{P'}$ of $\mathcal{P}$ with $\mu(\mathcal{P'}) = 1$ such that $\mathcal{P'}$ lies entirely on a line. 
			
		\end{enumerate}
	\end{assum}
	For each $p=[x,y,z]^T\in \mathcal{P}^o$, we modify the notation $f_i(p)$ and $h_n(p)$ in the 2-D scenario to suit the 3-D scenario here:
	\begin{align}
	\label{fip_3d}
		&f_{i}(p) \eq 
		\begin{bmatrix}
		    \frac{1}{\sigma_a}
		\arctan\left(\frac{y_i-y}{x_i-x}\right) \\
			\frac{1}{\sigma_e}
		\arctan\left(\frac{z_i-z}{\sqrt{(x_i-x)^2+(y_i-y)^2}}\right)
		\end{bmatrix},~i=1,\cdots,n
	\end{align}   and $h_n(p) = \frac{1}{n}\sum_{i=1}^n(f_i(p)-f_i(p^o))^T(f_i(p)-f_i(p^o))$.
	 We have the following results on the asymptotic localizability for  the 3-D AOA-based localization problem.
	\begin{thm}
	\label{lem_uniq_3d}
	    Under Assumptions \ref{assum_coordinates_3d}-\ref{assum_projection}, the  true signal source is asymptotically localizable.
	\end{thm}
	
		\subsection{Maximum likelihood estimator}
		
This section  derives the ML estimator for the  3-D AOA-based localization problem defined in \eqref{mea_3d} and  proves its consistency and asymptotic efficiency.

	Assumption \ref{assum_mea_noise_3d} entails that the log-likelihood function of the ML estimation for the problem \eqref{mea_3d} is
	\begin{align*}
		\ell_n&(p) = -n\ln(\sqrt{2\pi}\sigma_a)-n\ln(\sqrt{2\pi}\sigma_e)\\
		&-\frac{1}{2}\sum_{i=1}^n\Bigg[\frac{1}{\sigma_a^2}\left( a_i-\arctan\left(\frac{y_i-y}{x_i-x}\right) \right)^2 + \frac{1}{\sigma_e^2}\left( e_i-\arctan\left(\frac{z_i-z}{\sqrt{(x_i-x)^2+(y_i-y)^2}}\right) \right)^2\Bigg].
	\end{align*}
	Therefore, the ML estimation is given by
	\begin{align}
		\min_{p=[x,y,z]^T \in \mathbb{R}^3}\frac{1}{n}\sum_{i=1}^n\Bigg[ \frac{1}{\sigma_a^2}\left( a_i-\arctan\left(\frac{y_i-y}{x_i-x}\right) \right)^2 	+ \frac{1}{\sigma_e^2}\left( e_i-\arctan\left(\frac{z_i-z}{\sqrt{(x_i-x)^2+(y_i-y)^2}}\right) \right)^2\Bigg], \label{ML_3d}
	\end{align}
	which is equivalent to maximizing $\ell_n(p)$. Denote the ML estimator as $\widehat{p}_n^{\rm ML}$, which maximizes $\ell_n(p)$.
	
	Furthermore, let $\nabla f_{i}(p)$  be the the Jacobian of $f_{i}(p)$ with respect to $p$, i.e.,
	\begin{align}
	\label{nablafip_3d}
		&\nabla f_{i}(p) = \left[\frac{\partial f_{i}(p)}{\partial x}, ~ \frac{\partial f_{i}(p)}{\partial y},~ \frac{\partial f_{i}(p)}{\partial z}  \right]
		=\begin{bmatrix}
		    \frac{1}{\sigma_a}\frac{y_i-y}{(x_i-x)^2+(y_i-y)^2} & \frac{1}{\sigma_e}\frac{(x_i-x)(z_i-z)}{\sqrt{(x_i-x)^2+(y_i-y)^2}\|p_i-p\|^2}\\
		    -\frac{1}{\sigma_a}\frac{x_i-x}{(x_i-x)^2+(y_i-y)^2} & \frac{1}{\sigma_e}\frac{(y_i-x)(z_i-z)}{\sqrt{(x_i-x)^2+(y_i-y)^2}\|p_i-p\|^2}\\
		    0 & \frac{1}{\sigma_e} \frac{-\sqrt{(x_i-x)^2+(y_i-y)^2}}{\|p_i-p\|^2}
		\end{bmatrix}^T.
	\end{align}
Similar to the 2-D scenario, we have: 
		\begin{lem}
			\label{lem_convergeM_3d}
			Under Assumptions \ref{assum_coordinates_3d}-\ref{assum_projection}, we have:
			\begin{enumerate}[(i)]
				\item The matrix  
				$  \frac{1}{n}\sum_{i=1}^n \big(\nabla f_{i}(p)\big)^T \nabla f_{i}(p)$ converges uniformly for $p \in \mathcal{P}^o$  as $n \to \infty$.
				\item The limit $M^o \eq \lim_{n \to \infty} \frac{1}{n}\sum_{i=1}^n \big(\nabla f_{i}(p^o)\big)^T \nabla f_{i}(p^o)$ is nonsingular.
			\end{enumerate}
		\end{lem}
		
		We then present the consistency and asymptotic normality of  the ML estimator  in the following theorem, which is the counterpart of Theorem \ref{thm_consis_asymnormal_2d} in the 2-D scenario. 
		\begin{thm}
			\label{thm_consis_asymnormal_3d}
			Under Assumptions \ref{assum_mea_noise_3d}-\ref{assum_projection}, we have $\widehat{p}_n^{\rm ML} \to p^o$ almost surely as $n \to \infty$ with the asymptotic rate
			\begin{equation*}
				\sqrt{n}(\widehat{p}_n^{\rm ML} - p^o) \to \mathcal{N}\left(0,\left(M^o\right)^{-1}\right), \quad {\rm as} ~ n \to \infty.
			\end{equation*}
		\end{thm}
		Similar to the  2-D scenario,  the ML estimator $\widehat{p}_n^{\rm ML}$ is asymptotically efficient.

		\subsection{ Asymptotically efficient two-step estimator}
		
		Following the framework of the asymptotically efficient two-step estimator introduced in the 2-D scenario, we first propose a $\sqrt{n}$-consistent initial estimate of the source location and then refine it into an asymptotically efficient estimator
		via one-step GN iteration.
		
		\subsubsection{ $\sqrt{n}$-consistent estimator}
		
		For $i=1,...,n$, define
		\begin{align*}
			&r_i^o \eq \sqrt{(x_i-x^o)^2 + (y_i-y^o)^2},\\
			&d_i^o \eq \sqrt{(x_i-x^o)^2 + (y_i-y^o)^2 + (z_i-z^o)^2},\\
			&a_i^o \eq \arctan\left((y_i-y^o)/(x_i-x^o)\right),\\
			&e_i^o \eq \arctan\left((z_i-z^o)/\sqrt{(x_i-x^o)^2+(y_i-y^o)^2}\right).
		\end{align*}
		The measurement model \eqref{mea_3d} is equivalent to
		\begin{subequations}
			\label{linear_3d}
			\begin{align}
				\label{linear_3d_b}
				&(x_i-x^o) \sin(a_i) - (y_i-y^o) \cos(a_i) = r_i^o \sin(\varepsilon_i^a),\\
				\label{linear_3d_e}
				&	r_i^o\sin(e_i) - (z_i-z^o)\cos(e_i) = d_i^o\sin(\varepsilon_i^e).
			\end{align}
		\end{subequations}
		
		Moving terms of Eq. \eqref{linear_3d} produces the following linear regression form with respect to $p^o$
		\begin{subequations}
			\label{linear_3d_mat}
			\begin{align}
				\label{linear_3d_b_mat}
				h_i^T (p_i)_{1:2} &= h_i^T p^o_{1:2} + r_i^o\sin(\varepsilon_i^a),\\
				\label{linear_3d_e_mat}
				\sin(e_i)r_i^o-\cos(e_i)z_i &= -\cos(e_i)z^o + d_i^o\sin(\varepsilon_i^e), 
			\end{align} 
		\end{subequations}
		where $h_i = [\sin(a_i),  -\cos(a_i)]^T$, for $i=1,...,n$.
		
		Under Assumptions \ref{assum_mea_noise_3d}-\ref{assum_coordinates_3d}, Eq. \eqref{linear_3d_b_mat} is exactly Eq. \eqref{linear_2d_mat} in the 2-D scenario. 
		Therefore, according to \eqref{biaseli_esti_2d}, the BELS estimator of the first two coordinates $p^o_{1:2}$ of $p^o$  is given by
		\begin{align}
			(\widehat{p}_{1:2})_{n}^{\rm BE} &= \left(\frac{1}{n}X^TX - \mathbb{V}(\sin(\varepsilon_1^a)) I_2\right)^{-1} \left(\frac{1}{n}X^TY-\mathbb{V}(\sin(\varepsilon_1^a))\Big(\frac{1}{n}\sum_{i=1}^{n}(p_i)_{1:2}\Big)
			\right), 	\label{biaseli_esti_3d_b}
		\end{align}
		where the $i$-th row of matrix $X$ is $h_i^T$ and the $i$-th element of vector $Y$ is $h_i^T(p_i)_{1:2}$, for $i=1,...,n$. 
		The $\sqrt{n}$-consistency of $(\widehat{p}_{1:2})_{n}^{\rm BE}$ follows directly from Theorem \ref{thm_be_2d} in the 2-D scenario, which is presented in the following proposition.
		\begin{prop}
			\label{prop_be_3d_b}
			Under Assumptions \ref{assum_mea_noise_3d}-\ref{assum_projection}, the BELS estimator \eqref{biaseli_esti_3d_b} for the first two coordinates of the true source is $\sqrt{n}$-consistent, i.e., $(\widehat{p}_{1:2})_{n}^{\rm BE}- p^o_{1:2} = O_p(1/\sqrt{n})$.
		\end{prop}
		In what follows, we will derive the BELS estimator for the third coordinate $z^o$. By stacking \eqref{linear_3d_e_mat} for $n$ sensors, we obtain the following linear regression form with respect to $z^o$:
		\begin{equation}\label{zphi}
			\Gamma = \Phi z^o+\zeta,
		\end{equation}
		where the $i$-th element of vector $\Gamma$ is $\sin(e_i)r_i^o-\cos(e_i)z_i$, the $i$-th element of vector $\Phi$ is $-\cos(e_i)$, and the $i$-th element of vector $\zeta$ is $d_i^o\sin(\varepsilon_i^e)$, for $i=1,...,n$. Then, the resulting LS estimator is
		\begin{equation}
			\label{z_biased}
			\widehat{z}_n^{\rm B} = \left(\Phi^T\Phi\right)^{-1}\Phi^T\Gamma.
		\end{equation}
		The correlation between $\Phi$  and $\zeta$ and the unavailability of $\{r_i^o\}_{i=1}^n$ involved in $\Gamma$  make the LS estimator \eqref{z_biased} neither consistent nor applicable.
		Similar to the 2-D scenario, we rewrite   \eqref{linear_3d_e_mat} as
		\begin{align}
&\sin(e_i)r_i^o-\cos(e_i)z_i = -e^{-\sigma_e^2/2}\cos(e_i^o)z^o + \eta_i,\\
\nonumber
&\hspace{1mm}\eta_i = - (\cos(\varepsilon_i^e)-e^{-\sigma_e^2/2})\cos(e_i^o)z^o+\sin(\varepsilon_i^e)(d_i^o+\sin(e_i^o)z^o). \label{linear_3d_biaseli_e}
		\end{align}
		Its vector-matrix form  is
		\begin{equation}
			\label{gg}
			\Gamma = \Phi^o z^o + \eta,
		\end{equation}
		where the $i$-th element of vector $\Phi^o$ is $-e^{-\sigma_e^2/2}\cos(e_i^o)$ and the $i$-th element of vector $\eta$ is $\eta_i$.	For the regressor matrix $\Phi^o$, we have the following proposition.
		\begin{prop}
			\label{prop3}
			Under Assumptions \ref{assum_coordinates_3d}-\ref{assum_projection}, $\left(\Phi^o\right)^T\Phi^o/n$ is  bounded from below by a positive constant regardless of $n$.
		\end{prop}
	
		Thus, we  define the LS estimator of the model \eqref{gg} for $z^o$ by
		\begin{equation}
			\label{unbiased_esti_z}
			\widehat{z}_n^{\rm UB} = \left((\Phi^o)^T\Phi^o\right)^{-1}(\Phi^o)^T\Gamma,
		\end{equation}
		which shares  $\sqrt{n}$-consistency given in the following proposition.
		\begin{prop}
			\label{thm_ub_z}
			Under Assumptions \ref{assum_mea_noise_3d}-\ref{assum_projection}, the LS estimator $\widehat{z}_{n}^{\rm UB}$ is $\sqrt{n}$-consistent, i.e., $\widehat{z}_{n}^{\rm UB}-z^o = O_p(1/\sqrt{n})$.
		\end{prop}

		Although the $\sqrt{n}$-consistent LS estimator $\widehat{z}_{n}^{\rm UB}$ cannot be implemented in practice due to the inaccessibility of the unavailable $r_i^o$, the first two coordinates of the source involved in $r_i^o$ can be estimated using the BELS estimator $(\widehat{p}_{1:2})_n^{\rm BE}$, as given in \eqref{biaseli_esti_3d_b}.
		Write $(\widehat{p}_{1:2})_n^{\rm BE}$ in the element-wise form $(\widehat{p}_{1:2})_n^{\rm BE}= [\widehat{x}_n^{\rm BE}, ~ \widehat{y}_n^{\rm BE}]^T$ and define
		\begin{align}
			\widehat{r}_i = \sqrt{(x_i-\widehat{x}_n^{\rm BE})^2 + (y_i-\widehat{y}_n^{\rm BE})^2}\label{hatr}
		\end{align}
		for all $i=1,...,n$. Thus, we have
		\begin{align}
			\nonumber
			\widehat{r}_i &= \sqrt{(x_i-\widehat{x}_n^{\rm BE})^2 + (y_i-\widehat{y}_n^{\rm BE})^2}\\
			&=\sqrt{(r_i^o)^2 + O_p(1/\sqrt{n})}=r_i^o + O_p(1/\sqrt{n}),\label{rr}
		\end{align}
		where the second equation holds because of Proposition \ref{prop_be_3d_b} and Assumption \ref{assum_coordinates_3d}. Let $\widehat{\Gamma}$ be the vector with its $i$-th element being $\sin(e_i)\widehat{r}_i-\cos(e_i)z_i$ for $i=1,...,n$. Then, we propose the following BELS estimator for $z^o$:
		\begin{equation}
			\label{biaseli_esti_z}
			 \!\!\!\widehat{z}_n^{\rm BE} \!= \! \left(\frac{1}{n}\Phi^T\Phi \!-\! \mathbb{V}(\sin(\varepsilon_1^e))\!\right)^{-1}\!\!\left(\frac{1}{n}\Phi\widehat{\Gamma} \!- \! \mathbb{V}(\sin(\varepsilon_1^e))\bar{z}\!\right),
		\end{equation}
		where $\mathbb{V}(\sin(\varepsilon_1^e))= (1-e^{-2\sigma_e^2})/2$ 
		and $\bar{z}=\frac1{n}\sum_{i=1}^{n}z_i$.
		\begin{thm} \label{BELS_3D_case}
			Under Assumptions \ref{assum_mea_noise_3d}-\ref{assum_projection}, the BELS estimator $\widehat{z}_n^{\rm BE}$ is $\sqrt{n}$-consistent, i.e., $\widehat{z}_n^{\rm BE} - z^o = O_p(1/\sqrt{n})$.
		\end{thm}
		
		Thus, we can define the BELS estimator of the true source $p^o$ by
		\begin{equation}
			\label{bels_3d}
			\widehat{p}_n^{\rm BE} = \begin{bmatrix}
				(\widehat{p}_{1:2})_n^{\rm BE}	\\
				\widehat{z}_n^{\rm BE}	
			\end{bmatrix},
		\end{equation}
		where $(\widehat{p}_{1:2})_n^{\rm BE}$ and $\widehat{z}_n^{\rm BE}$ are given in \eqref{biaseli_esti_3d_b} and \eqref{biaseli_esti_z}, respectively. Combining with Proposition \ref{prop_be_3d_b} and Theorem \ref{BELS_3D_case}, we obtain the $\sqrt{n}$-consistency of $\widehat{p}_n^{\rm BE}$, which is stated in the following theorem.
		\begin{thm}
		Under Assumptions \ref{assum_mea_noise_3d}-\ref{assum_projection}, the BELS estimator $\widehat{p}_n^{\rm BE}$  given by \eqref{bels_3d} for the true source $p^o$
			is $\sqrt{n}$-consistent, i.e., $\widehat{p}_n^{\rm BE} - p^o = O_p(1/\sqrt{n})$.
		\end{thm}
		When the noise variances $\sigma_a^2$ and $\sigma_e^2$ are unknown, the BELS estimator $\widehat{p}_{n}^{\rm BE}$ in \eqref{bels_3d} remains $\sqrt{n}$-consistent if $\mathbb{V}(\sin(\varepsilon_1^a)) $ and $\mathbb{V}(\sin(\varepsilon_1^e))$ are replaced by their respective $\sqrt{n}$-consistent estimators. The $\sqrt{n}$-consistent estimator for $\mathbb{V}(\sin(\varepsilon_1^a))$ can refer to \eqref{sina} for the 2-D scenario. For completing the BELS estimator \eqref{bels_3d}, we aim to derive a $\sqrt{n}$-consistent estimator of $\mathbb{V}(\sin(\varepsilon_1^e))$ here, which is similar to that given for $\mathbb{V}(\sin(\varepsilon_1^a))$ in the 2-D scenario following the idea. Denote
		\begin{equation}
			\label{RU}
		\!\!	R_n\! \eq \frac{1}{n}\!\begin{bmatrix}
				\Phi^T \\
				\widehat{\Gamma}
			\end{bmatrix}\!\!
			\begin{bmatrix}
				\Phi,  \widehat{\Gamma}
			\end{bmatrix}, ~
			 U_n \eq \begin{bmatrix}
				1 & \bar{z} \\
				\bar{z} & \frac{1}{n}\sum\limits_{i=1}^n \big(z_i^2+(\widehat{r}_i)^2\big)
			\end{bmatrix}.
		\end{equation}
		Thus, we estimate $\mathbb{V}(\sin(\varepsilon_1^e)) $ by
		\begin{equation}
			\label{W}
			\widehat{v}^e_n \eq \frac{1}{\lambda_{\rm max}(R_n ^{-1}U_n)},
		\end{equation}
		which is a $\sqrt{n}$-consistent estimator of $\mathbb{V}(\sin(\varepsilon_1^e))$ given in
		the following theorem.
		\begin{thm}
			\label{thm_var_3d}
			Under Assumptions \ref{assum_mea_noise_3d}-\ref{assum_projection}, $\widehat{v}^e_n$ is a $\sqrt{n}$-consistent estimator of $\mathbb{V}(\sin(\varepsilon_1^e))$, i.e., $\widehat{v}^e_n - \mathbb{V}(\sin(\varepsilon_1^e)) = O_p(1/\sqrt{n})$.
		\end{thm}

		We summarize the estimation procedure of the variance of the sine of the elevation angle noise in Algorithm \ref{algo3}.
		
		\begin{algorithm}[!h]
			\caption{The estimation algorithm of the variance of the sine of the elevation angle noise }  
			\label{algo3}
			\begin{algorithmic}[1]
				\Require Sensor locations $\{p_i\}_{i=1}^n$ and AOA  measurements $\{a_i,e_i\}_{i=1}^n$.
				\State Calculate the BELS estimate $(\widehat{p}_{1:2})_n^{\rm BE}$ of the first two coordinates of the source location via Algorithm \ref{algo2};
				\State Calculate $\widehat{\Gamma}$ according to $(\widehat{p}_{1:2})_n^{\rm BE}$ and \eqref{hatr};
				\State Calculate $R_n$ and $U_n$ according to \eqref{RU};
				\State Calculate the maximum eigenvalue of $R_n^{-1}U_n$;
				\Ensure \!The estimate for $\mathbb{V}(\sin(\varepsilon_1^e))$: $ \widehat{v}^e_n \!= 1/\lambda_{\rm max}(R_n^{-1}\!U_n)$.
			\end{algorithmic}
		\end{algorithm}

		\subsubsection{Gauss-Newton refinement}
		
		In this subsubsection, we apply the GN iterations to the $\sqrt{n}$-consistent initial estimator derived above, yielding an asymptotically efficient estimator for the 3-D AOA-based localization problem.
		
		
		Let 
			$\widehat{p}_n^{\rm BE}$ be the consistent estimator \eqref{bels_3d} for $p^o$ derived in the first step.
		Thus, the one-step GN iteration    is
\begin{subequations}\label{gn_3d}
\begin{align} 
&\widehat{p}_n^{\rm GN} = \widehat{p}_n^{\rm BE} +  \left(J_n^T J_n\right)^{-1}J_n^T(\rho - f),\\
&J_n = \left[\nabla f_1(\widehat{p}_n^{\rm BE})^T,  \cdots , \nabla f_n(\widehat{p}_n^{\rm BE})^T\right]^T,\\
        &f = [f_1(\widehat{p}_n^{\rm BE})^T,  \cdots , f_n(\widehat{p}_n^{\rm BE})^T]^T, \\
    &\rho = [a_1, e_1, \cdots, a_n, e_n]^T,
\end{align}
where $f_i(\cdot)$ and $\nabla f_i(\cdot)$ are defined by \eqref{fip_3d}  and \eqref{nablafip_3d}, respectively. 
\end{subequations}
		Similarly, by \cite{Mu2017}, \cite{Lehmann1998}, the one-step GN iteration estimator $\widehat{p}_n^{\rm GN}$ given by \eqref{gn_3d} is   asymptotically efficient for the 3-D AOA-based localization problem  \eqref{mea_3d}.

		\subsubsection{Complete procedure of asymptotically efficient two-step estimator}
				The whole procedure of the proposed two-step estimator in the 3-D scenario is summarized in Algorithm \ref{algo4}. 
		\begin{algorithm}[!h]
			\caption{The estimation algorithm for consistent and asymptotically efficient two-step estimator   using AOA  measurements (3-D scenario)}  
			\label{algo4}
			\begin{algorithmic}[1]
				\Require Sensor locations $\{p_i\}_{i=1}^n$, AOA  measurements $\{a_i,e_i\}_{i=1}^n$, the azimuth angle noise variance $\sigma_a^2$ (if available), and the elevation angle noise variance $\sigma_e^2$ (if available).
				\State Apply Algorithm \ref{algo2} to obtain the BELS estimate $(\widehat{p}_{1:2})_n^{\rm BE} = [\widehat{x}_n^{\rm BE}, \widehat{y}_n^{\rm BE}]^T$ of the first two coordinates of the source location;
				\If {$\sigma_e^2$ is available}
				\State Calculate the variance of the sine of noise by $\sigma_e^2$;
				\State Calculate the BELS estimate $\widehat{z}_n^{\rm BE}$ of the third coordinate of source location according to \eqref{biaseli_esti_z};
				\Else
				\State Apply Algorithm \ref{algo3} to obtain the estimate $\widehat{v}^e_n$ for the variance of the sine of the elevation angle noise;
				\State Calculate the BELS estimate $\widehat{z}_n^{\rm BE}$ of the third coordinate of source location according to \eqref{biaseli_esti_z} with $\mathbb{V}(\sin(\varepsilon_1^e))$ being replaced by $\widehat{v}^e_n$;
				\EndIf
				\State Set $\widehat{p}^{\rm BE}_n = \left[\widehat{x}_n^{\rm BE},~ \widehat{y}_n^{\rm BE}, ~ \widehat{z}_n^{\rm BE}\right]^T$;
				\State Run   one-step GN iteration \eqref{gn_3d} and obtain $\widehat{p}_n^{\rm GN}$;
				\Ensure The source location estimate $\widehat{p}_n^{\rm GN}$.
			\end{algorithmic}
		\end{algorithm}
		
		In Algorithm \ref{algo3}, Line 1 has a time complexity of $\mathcal{O}(n)$, based on the time complexity of Algorithm \ref{algo2}. The remaining steps in Algorithm \ref{algo3} are similar to those in Algorithm \ref{algo1}, which also has a time complexity of $\mathcal{O}(n)$. Therefore, the overall time complexity of Algorithm \ref{algo3} is $\mathcal{O}(n)$. Similarly, in Algorithm \ref{algo4}, Line 1 invokes Algorithm \ref{algo2}, which has a time complexity of $\mathcal{O}(n)$. The rest of Algorithm \ref{algo4} follows the same structure as Algorithm \ref{algo2}, also requiring $\mathcal{O}(n)$ time. Hence, the overall time complexity of Algorithm \ref{algo4} is $\mathcal{O}(n)$.
				\begin{figure*}[!t]
			\centering
			\begin{subfigure}[b]{.32\textwidth}
			\centering
			 \includegraphics[width=\textwidth]{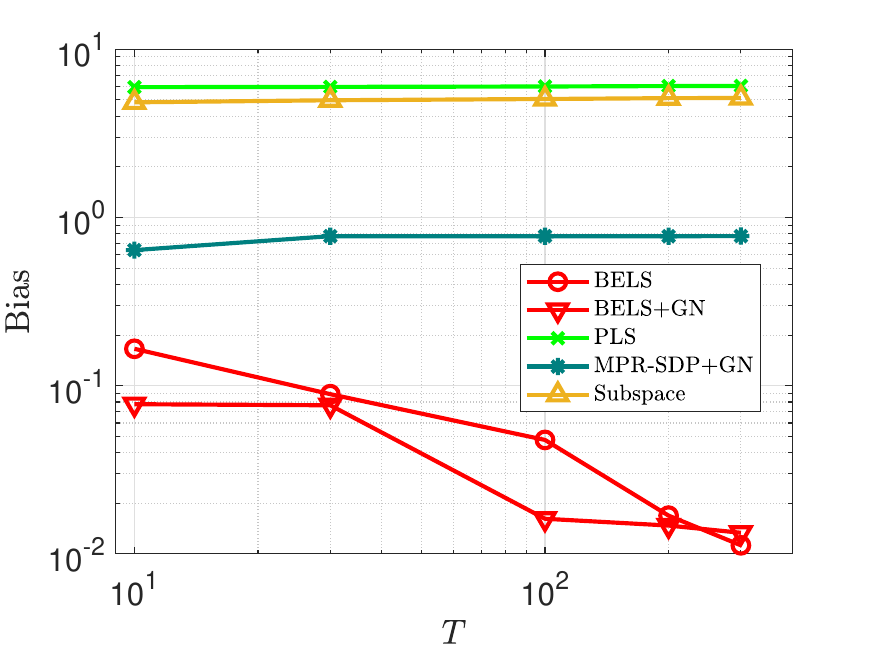}
			\caption{Biases for varying numbers of measurements}
 			\label{fig:2d_bias_varyingT}
			\end{subfigure}
			\begin{subfigure}[b]{.32\textwidth}
				\centering				\includegraphics[width=\textwidth]{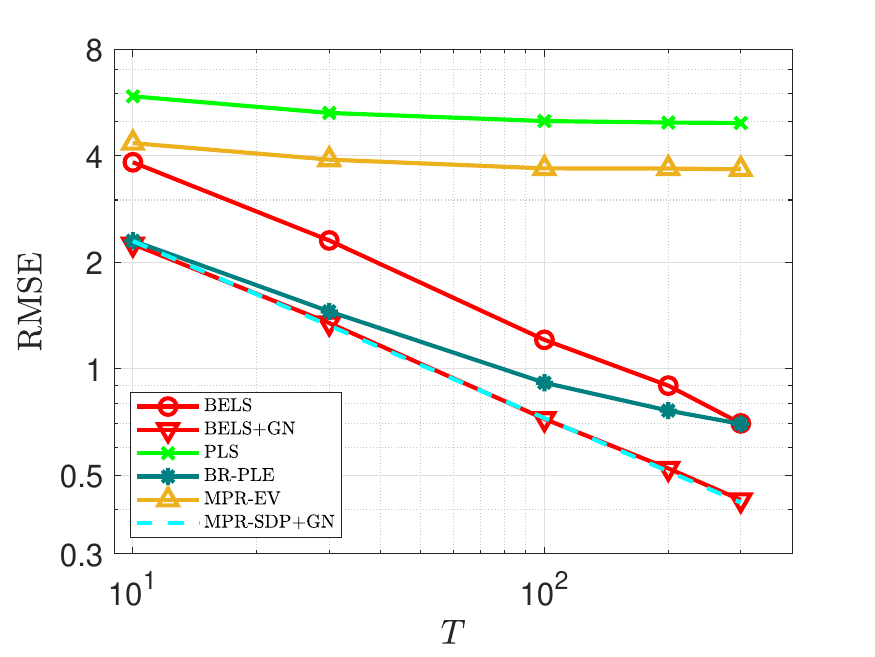}
				\caption{RMSEs for varying numbers of measurements}
				\label{fig:2d_rmse_varyingT}
			\end{subfigure}
			\begin{subfigure}[b]{.32\textwidth}
				\centering
	\includegraphics[width=\textwidth]{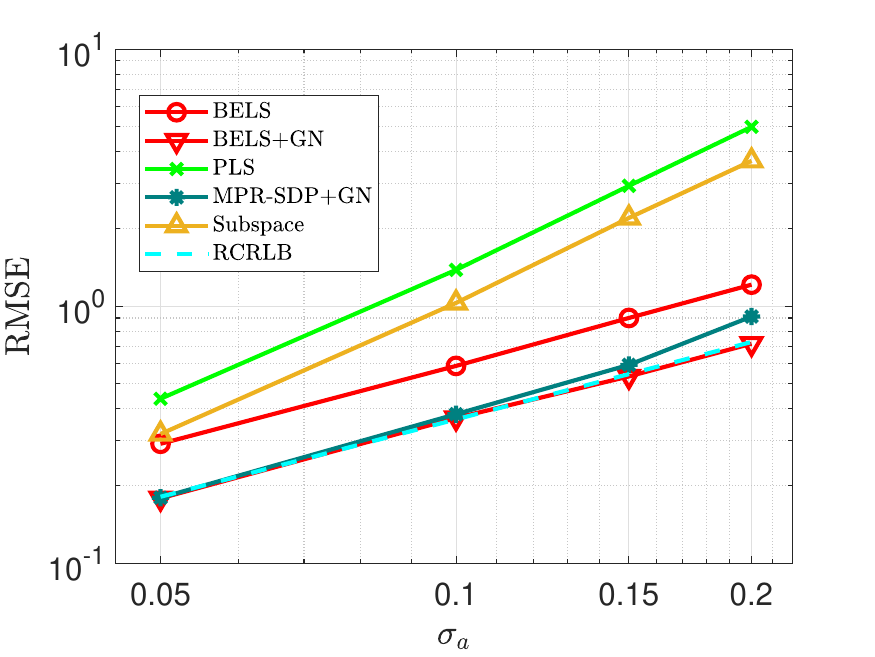}
				\caption{RMSEs for varying noise intensities}
				\label{fig:2d_rmse_varyingsigma}
			\end{subfigure}
			\caption{2-D: biases and RMSEs of the estimators for fixed sensors.}
		\end{figure*}
		\section{Simulations}
		\label{sec4}
In this section, we perform Monte Carlo simulations to validate the theoretical results of the proposed asymptotically efficient two-step estimator for AOA-based localization problems. Simulations are carried out separately for both 2-D and 3-D scenarios. 

In each scenario, we conduct \( N \) Monte Carlo trials using independent realizations of the measurement noise. For a given estimator, let \( \widehat{p}_{n,j} \) denote the estimate obtained in the \( j \)-th trial.
 Its performance is assessed in terms of the bias and root mean squared error (RMSE), defined as \cite{Wang2023, Zeng2024}:
\begingroup
\allowdisplaybreaks
\begin{align*}
    &\mathrm{Bias}(\widehat{p}_n) = \sum_{i=1}^m \big|\big[\Delta(\widehat{p}_n)\big]_i\big|, \quad \Delta(\widehat{p}_n) =   \frac{1}{N} \sum_{j=1}^N \widehat{p}_{n,j} - p^o  , \\
    &\mathrm{RMSE}(\widehat{p}_n) = \sqrt{\frac{1}{N} \sum_{j=1}^N \| \widehat{p}_{n,j} - p^o \|^2},
\end{align*}
\endgroup
where \( p^o \) denotes the true source, and \( m \) is the dimension of the source coordinates, with \( m = 2 \) for the 2-D case and \( m = 3 \) for the 3-D case.
We employ the root Cramér–Rao lower bound (RCRLB) as a performance benchmark to assess whether the proposed estimators are asymptotically efficient.

Our proposed estimators are denoted as follows:  
\begin{itemize}
    \item When the noise variance is known, we refer to the first-step estimator as BELS and its refined version (after the second step) as BELS+GN.
    
    \item When the variance is unknown, we use the estimated variance \(\widehat{v}_n^a\) for the 2-D case (\(\widehat{v}_n^a\) and \(\widehat{v}_n^e\) for the 3-D case) and label the corresponding estimators as  BELS\((\widehat{v}_n^a)\) and BELS\((\widehat{v}_n^a)\)+GN for the 2-D case (BELS\((\widehat{v}_n^a, \widehat{v}_n^e)\) and BELS\((\widehat{v}_n^a, \widehat{v}_n^e)\)+GN for the 3-D case), respectively.
\end{itemize}
		In 2-D scenario, we compare our methods with the following existing estimators:
		\begin{enumerate}[(i)]
			\item PLS: the LS estimator \eqref{p_biased} that corresponds to the linear regression model \eqref{eq8} \cite{Lingren1978,Nardone1984};
			\item MPR-SDP+GN: the 2-D version of the modified polar representation, which is produced by the GN iteration  initialized with the SDP-based method \cite{Wang2018a};
			\item Subspace: a relaxation-based linear estimator reducing errors by making use of all AOA  and inter-sensor angular geometric information \cite{Luo2019}.
		\end{enumerate}
		In 3-D scenario, we compare our methods with the following existing estimators:
		\begin{enumerate}[(i)]
			\item PLS: the combination of the LS estimators \eqref{p_biased} and \eqref{z_biased}, which is generalized from the 2-D version;
			
			\item BR-PLE: the bias reduced solution to a CWLS problem that approximates the sine and cosine functions using their first-order Taylor expansions and then expands the parameter space by one dimension while imposing a norm constraint \cite{Wang2015};
			
			\item MPR-EV: the modified polar representation given by computing an eigenvector with bias reduction \cite{Sun2020};
			
			\item MPR-SDP+GN: the modified polar representation  produced by the GN iteration,  initialized using the SDP-based method \cite{Wang2018}.
		\end{enumerate}
		 All estimators are implemented in MATLAB and executed on an
AMD EPYC 7543 32-Core Processor.

		\subsection{2-D scenario: fixed sensors}
		\label{subsec:A}
		We deploy 10 fixed sensors at the following 2-D coordinates: $ p_1 = [0,100]^T,$ $p_2=[0,50]^T,$ $p_3=[50,50]^T,$ $p_4=[50,0]^T,$ $p_5=[50,-50]^T, $ $p_6 = [0,-50]^T,$ $p_7=[0,-100]^T,$ $p_8=[-50,-50]^T,$ $p_9=[-50,0]^T,$ $p_{10}=[-50,50]^T$, and  place the true source at $p^o=[60,10]^T$.
		Each sensor collects \(T\)  i.i.d. AOA  measurements, resulting in a total of \(n = 10T\) i.i.d. observations across all 10 sensors. As \(T\) increases, the growing number of measurements allows the asymptotic properties of the estimators to emerge clearly.  
This setup is statistically equivalent to deploying \(T\) identical, co-located sensors at each of the 10 fixed positions. As \(T \to \infty\), the empirical distribution of the sensor locations converges almost surely to a discrete probability distribution \(\mu\) defined by  $\mu(p) = 1/10$ for $p$ at one of the 10 placement sites, and $\mu(p)=0$, otherwise.
Consequently, the support of \(\mu\) consists precisely of these 10 points, and they are not collinear (as can be verified from their coordinates), and hence Assumptions \ref{assum_coordinates_2d} and \ref{assum_colinearity} hold.

We first examine the bias and RMSE of each estimator as functions of \(T\), the number of i.i.d. AOA  measurements per sensor. The angular measurement noise is a zero-mean Gaussian with standard deviation \(\sigma_a = 0.2\) rad. We evaluate performance across \(T \in \{10,\, 30,\, 100,\, 200,\, 300\}\), conducting 1,000 independent Monte Carlo trials for each \(T\) to ensure statistical reliability.

		
		 		\begin{figure}[!t]
		\centering
			\begin{subfigure}[b]{.43\textwidth}
				\centering				\includegraphics[width=\textwidth]{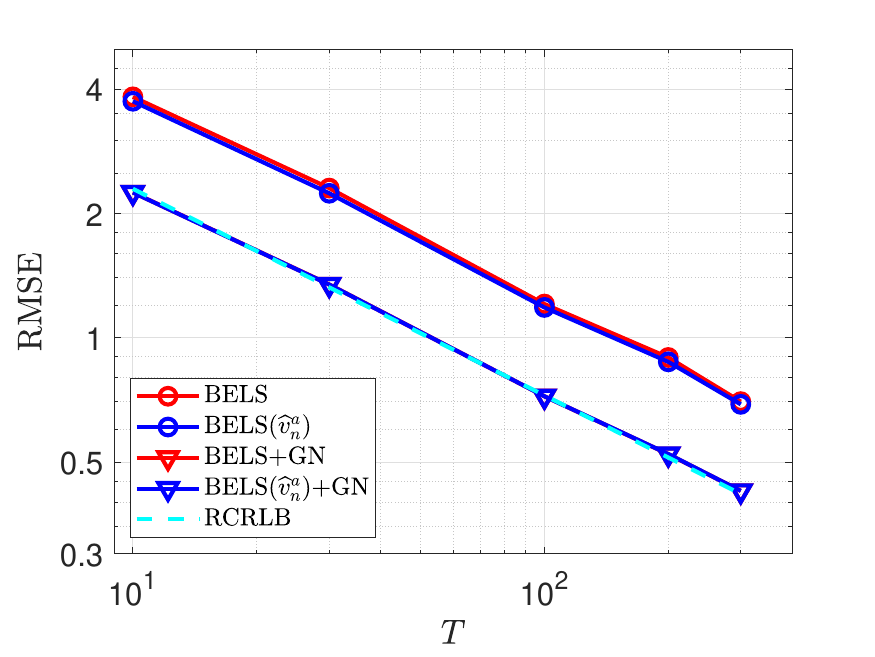}
				\caption{Varying numbers of measurements.}
		\label{fig:2d_compareVar_varyingT}
			\end{subfigure}
			\begin{subfigure}[b]{.43\textwidth}
				\centering
		\includegraphics[width=\textwidth]{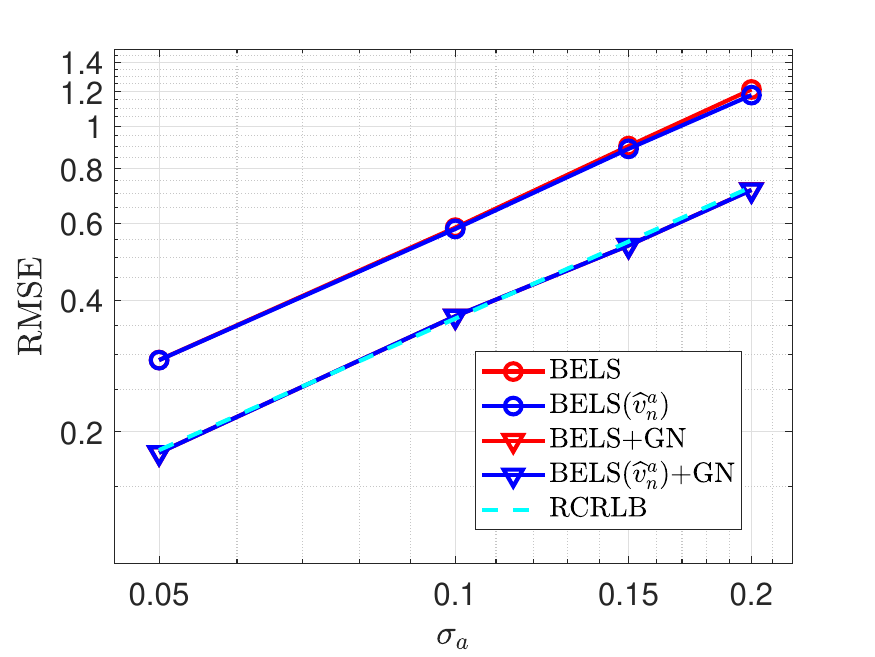}
				\caption{Varying noise intensities.}
			\label{fig:2d_compareVar_varyingsigma}
			\end{subfigure}
			\caption{2-D: RMSE comparison of  BELS and BELS+GN estimators between the true variance of the sine of the noise and the estimated one.}
			\label{fig:2d_compareVar}
		\end{figure}
		
						\begin{figure}[!h]
		\centering
			\begin{subfigure}[b]{.43\textwidth}
				\centering				\includegraphics[width=\textwidth]{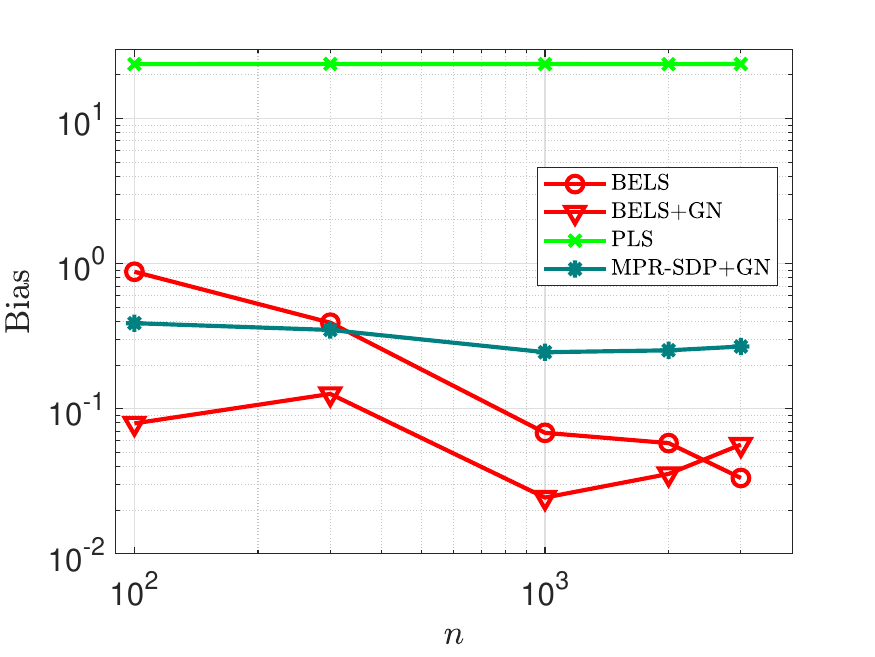}
				\caption{Biases for varying numbers of measurements.}
				\label{fig:2d_bias_random}
			\end{subfigure}
			\begin{subfigure}[b]{.43\textwidth}
				\centering
			\includegraphics[width=\textwidth]{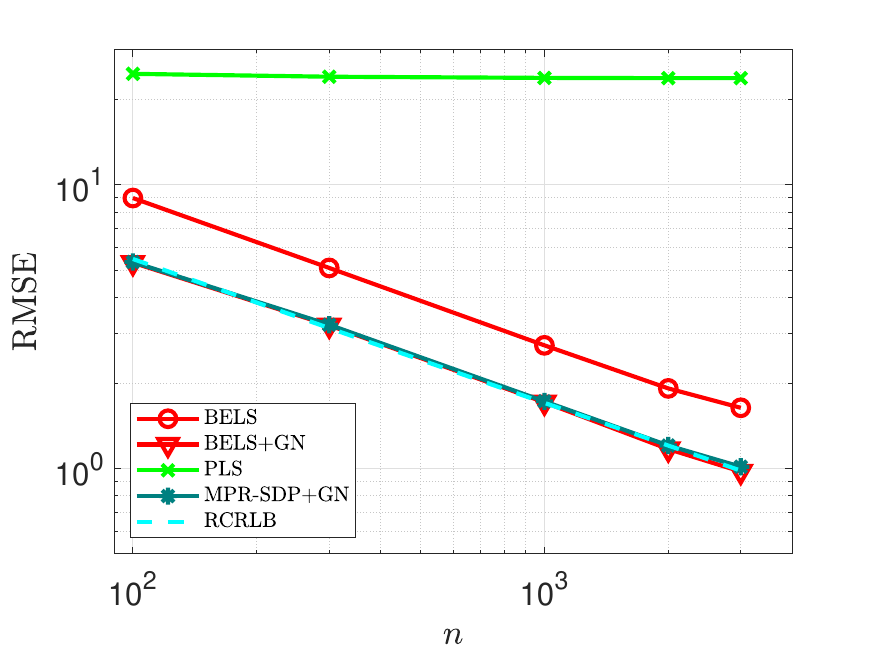}
				\caption{RMSEs for varying numbers of measurements.}
				\label{fig:2d_rmse_random}
			\end{subfigure}
			\caption{2-D: Biases and RMSEs of the estimators for random sensors.}
			\label{fig:2d_random}
		\end{figure}
		
Fig.~\ref{fig:2d_bias_varyingT} displays the biases of all estimators  as a function of \(T\). As \(T\) increases, the biases of the BELS and BELS+GN estimators decay toward zero subject to minor fluctuations, confirming their asymptotic unbiasedness. In contrast, all the PLS, Subspace, and  MPR-SDP+GN   suffer from a substantial and nonvanishing bias, indicating their asymptotic biasedness.


Fig.~\ref{fig:2d_rmse_varyingT} presents the RMSEs of all  estimators as functions of  \(T\). The proposed BELS+GN estimator achieves asymptotic efficiency, with its RMSE approaching the  RCRLB as \(T\) increases. The MPR-SDP+GN estimator exhibits comparable RMSE to BELS+GN for smaller \(T=10\), but its RMSE deviates from the  RCRLB for   larger \(T=30,100,200,300\), confirming its asymptotic inefficiency. 
Meanwhile, both the PLS and Subspace estimators exhibit significantly higher RMSEs than  the RCRLB due to their nonvanishing bias.

Moreover, we examine the  RMSEs of all estimators under varying noise intensities. Fig. \ref{fig:2d_rmse_varyingsigma} presents the RMSEs of all estimators under varying noise intensities $\sigma_a = 0.05, 0.1, 0.15$ and $0.2$ rad with  $T = 100$, where each point is the average of 1000 Monte-Carlo runs. 
In low-noise regions, all estimators exhibit small RMSEs, with both BELS+GN and MPR-SDP+GN attaining the RCRLB. As noise increases, only BELS+GN maintains RCRLB attainment, while other estimators show significantly higher RMSEs.
		
		\begin{table}[ht]
			\centering
			\caption{RMSEs of the estimates for the variance of the sine of noises}
			\resizebox{0.65\columnwidth}{!}{
				\begin{tabular}{cccccc}
					\hline
					\hline
					$n=100$  & $n=300$  & $n=1000$  & $n=2000$ & $n=3000$  \\ \hline
					0.00610 & 0.00339 & 0.00199 & 0.00139 & 0.00115   \\ \hline
			\end{tabular}}
			\label{table:2d_rmse_var}
		\end{table}

Finally, we evaluate the performance of the BELS\((\widehat{v}_n^a)\) and BELS\((\widehat{v}_n^a)\)+GN estimators in the case of unknown noise variance.
Table \ref{table:2d_rmse_var} confirms the $\sqrt{n}$-consistency of the proposed estimator for the variance of the sine of noises under $\sigma_a = 0.2$ rad.
Fig.~\ref{fig:2d_compareVar_varyingT} shows that BELS\((\widehat{v}_n^a)\) and BELS\((\widehat{v}_n^a)\)+GN achieve RMSEs virtually indistinguishable from those of BELS and BELS+GN across varying sample sizes.
Fig.~\ref{fig:2d_compareVar_varyingsigma} confirms this equivalence holds across a range of noise levels (\(\sigma_a = 0.05\)–\(0.2\) rad).  
Collectively, these results verify that BELS\((\widehat{v}_n^a)\) is \(\sqrt{n}\)-consistent, and BELS\((\widehat{v}_n^a)\)+GN is asymptotically efficient, achieving the RCRLB without requiring prior knowledge of the noise variance.

		\subsection{2-D scenario: random sensors}
This subsection considers a random sensor deployment scenario, in which the sensors are independently and uniformly distributed on a circle of radius 100 centered at the origin. Specifically, the location of the \(i\)-th sensor is generated as  
$
p_i = \big[100\cos(\beta_i),~ 100\sin(\beta_i)\big]^T$
 with  $\beta_i$ uniformly drawn from the uniform distribution $\mathcal{U}[0, 2\pi)$.
 The true source is located at $p^o=[150,0]^T$.
 The angular measurement noise is a zero-mean Gaussian with standard deviation \(\sigma_a = 0.2\) rad.
The number of sensors $n$ varies as $100,300,1000,2000,3000$ and each sensor collects one observation.
Thus,  Assumptions \ref{assum_coordinates_2d}–\ref{assum_colinearity} are satisfied under this configuration.

		\begin{figure*}[!t]
			\centering
			\begin{subfigure}[b]{.32\textwidth}
			    			\centering
			\includegraphics[width=\textwidth]{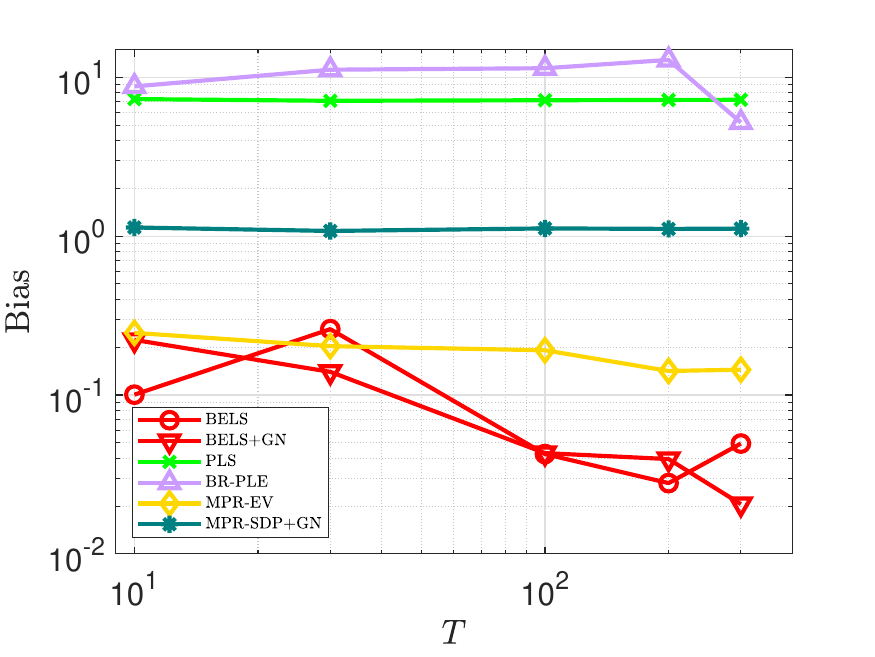}
			\caption{Biases for varying numbers of
measurement}
			\label{fig:3d_bias_varyingT}
			\end{subfigure}
			\begin{subfigure}[b]{.32\textwidth}
				\centering				\includegraphics[width=\textwidth]{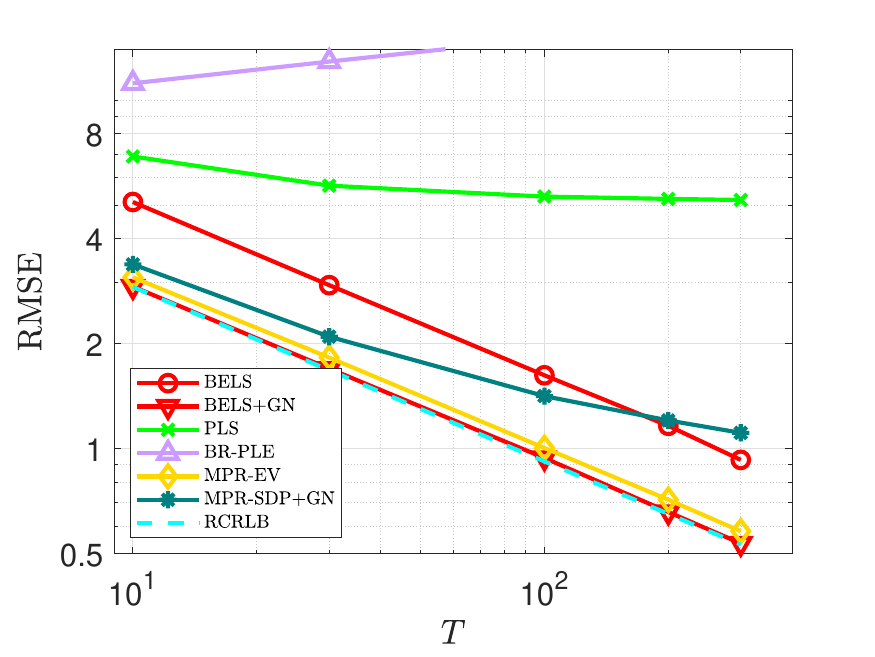}
				\caption{RMSEs for varying numbers of measurements}
				\label{fig:3d_rmse_varyingT}
			\end{subfigure}
			\begin{subfigure}[b]{.32\textwidth}
				\centering
			\includegraphics[width=\textwidth]{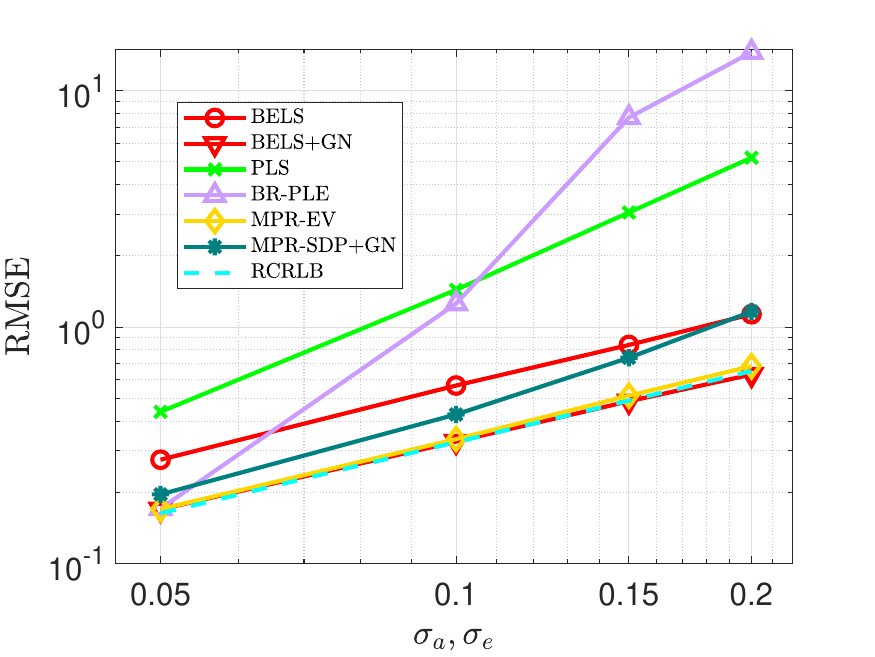}
				\caption{RMSEs for varying noise intensities}
				\label{fig:3d_rmse_varyingsigma}
			\end{subfigure}
			\caption{3-D: biases and RMSEs of the estimators for fixed sensors.}
		\end{figure*}
We evaluate the bias and RMSE of all estimators as a function of the total number of measurements \(n\), averaging over 1,000 Monte Carlo runs. The Subspace estimator is excluded from this simulation due to its high computational complexity in large-\(n\) regimes.
Fig.~\ref{fig:2d_bias_random} shows that both BELS and BELS+GN have a smaller bias than MPR-SDP+GN except for BELS with $n=100,300$ under random sensor deployment. While PLS suffers from a large nonvanishing bias.  
Fig.~\ref{fig:2d_rmse_random} further demonstrates that BELS+GN achieves the RCRLB across all tested \(n\), verifying its asymptotic efficiency. 
MPR-SDP+GN approaches the RCRLB at moderate \(n=100,300,1000\) and exhibits a slightly deviation from the RCRLB  for larger \(n=2000,3000\), reflecting its possible bias and sensitivity to initialization. In contrast, PLS yields significantly higher RMSE than the RCRLB due to its nonvanishing bias, underscoring its inefficiency.

		\subsection{3-D scenario: fixed and noncoplanar sensors}
		We deploy 10 fixed sensors at the following 3-D coordinates:
$p_1 = [50,50,50]^T,$ $p_2=[50,0,50]^T,$
$p_3=[50,50,-50]^T,$
$p_4=[50,100,0]^T,$
$p_5=[50,-50,50]^T,$
$p_6 = [-50,0,-50]^T,$
$p_7=[-50,-50,50]^T,$
$p_8=[-50,-50,-50]^T,$
$p_9=[-50,-100,0]^T,$
$p_{10}=[-50,50,-50]^T$,
		and place the true  source at $p^o=[60,10,10]^T$.
		Similar to the 2-D scenario in Section \ref{subsec:A}, each sensor makes $T$ rounds of i.i.d. observations with   $\sigma_a = \sigma_e = 0.2$ rad.
		Thus, Assumptions \ref{assum_coordinates_3d} and \ref{assum_projection} hold.
		For each $T$, we conduct 1000 Monte Carlo runs and the displayed results are based on the average of 1000 Monte Carlo runs.

Figs. \ref{fig:3d_bias_varyingT} and \ref{fig:3d_rmse_varyingT} present the biases and RMSEs of all estimators for varying $T=10,\, 30,\, 100,\, 200,\, 300$, respectively. 
We observe that: (i) BELS, BELS+GN, and MPR-EV all exhibit a diminishing trend in bias as \(T\) increases. Notably, both BELS and BELS+GN consistently achieve lower bias than MPR-EV except for BELS at \(T = 30\). In contrast, PLS, BR-PLE, and MPR-SDP+GN display nonvanishing biases that persist with increasing \(T\), confirming their asymptotic biasedness.
(ii) 
The BELS+GN estimator is asymptotically efficient, with its RMSE converging to the RCRLB as \(T \to \infty\). 
MPR-EV exhibits a decreasing RMSE as $T$ increases, confirming its consistency; however, its RMSE remains consistently above the RCRLB, indicating suboptimal asymptotic efficiency.
In contrast, PLS, BR-PLE, and MPR-SDP+GN converge to nonzero RMSE values,  confirming their asymptotic inefficiency.
		\begin{figure}[!h]
		\centering
			\begin{subfigure}[b]{.43\textwidth}
				\centering
			\includegraphics[width=\textwidth]{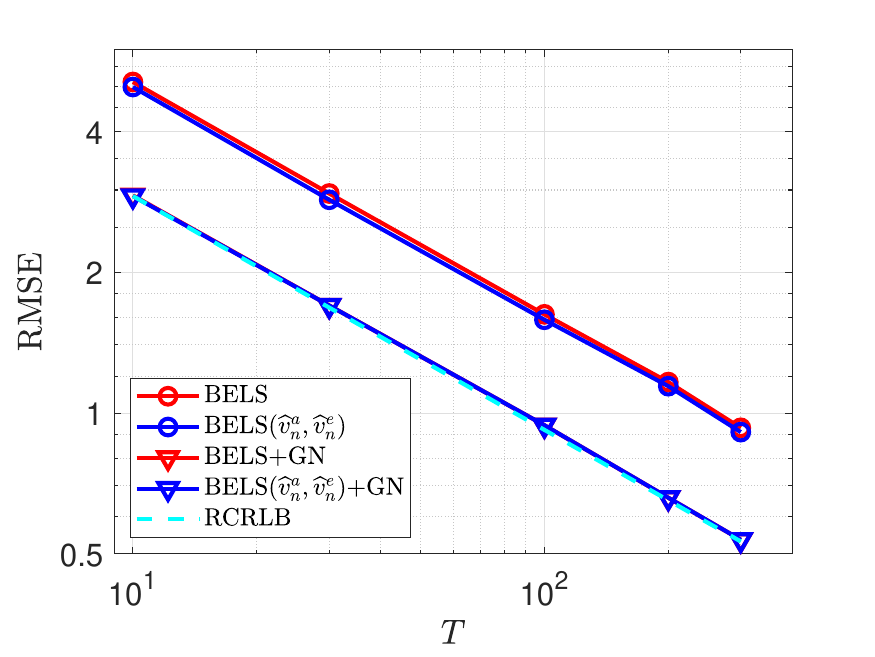}
				\caption{Varying numbers of measurements.}
				\label{fig:3d_compareVar_varyingT}
			\end{subfigure}
			\begin{subfigure}[b]{.43\textwidth}
				\centering
			\includegraphics[width=\textwidth]{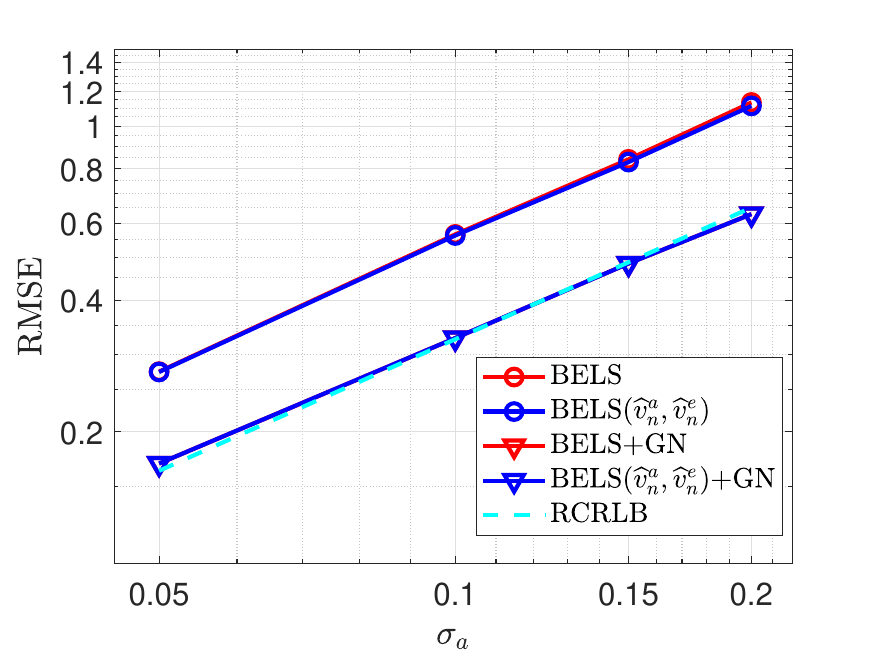}
				\caption{Varying noise intensities}
			\label{fig:3d_compareVar_varyingsigma}
			\end{subfigure}
			\caption{3-D: RMSE comparison of  BELS and BELS+GN estimators between the true variance of the sine of the noises and the estimated ones.}
			\label{fig:3d_compareVar}
		\end{figure}

Fig. \ref{fig:3d_rmse_varyingsigma} presents the RMSEs under varying noise intensities with sample sizes $T=200$ for $\sigma_a=\sigma_e=0.05, 0.1, 0.15$, and $0.2$ rad.  In low-noise intensities $\sigma_a=\sigma_e= 0.05,0.1$, both BELS+GN and MPR-EV achieve the RCRLB. As noise increases $\sigma_a=\sigma_e= 0.15,0.2$, only BELS+GN maintains RCRLB attainment, while all other estimators deviate from it.

				\begin{figure*}[!t]
		\centering
			\begin{subfigure}[b]{.24\textwidth}
				\centering
			\includegraphics[width=\textwidth]{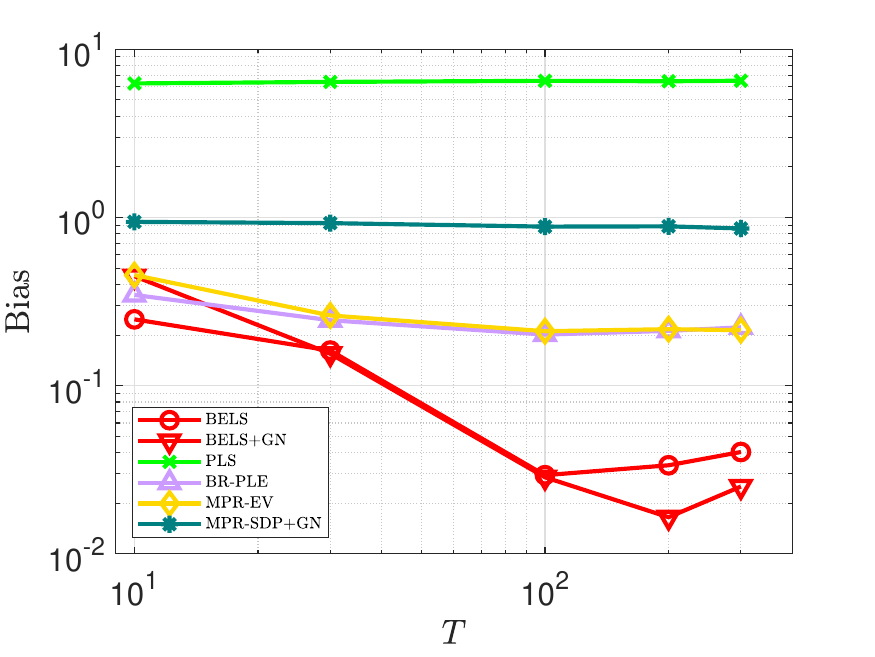}
				\caption{The source is not coplanar with the sensors.}
			\label{fig:3d_noncoplanar_source_bias}
			\end{subfigure}
			\begin{subfigure}[b]{.24\textwidth}
				\centering
			\includegraphics[width=\textwidth]{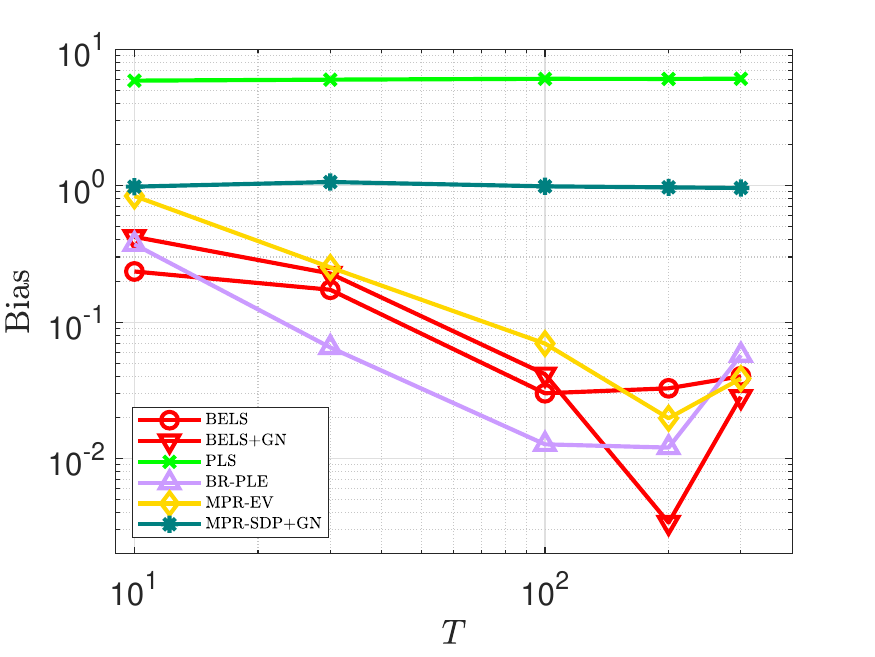}
				\caption{The source and sensors are coplanar.\\~}
			\label{fig:3d_coplanar_source_bias}
			\end{subfigure}
			\label{fig:3d_coplanar_bias}
		\centering
			\begin{subfigure}[b]{.24\textwidth}
				\centering				\includegraphics[width=\textwidth]{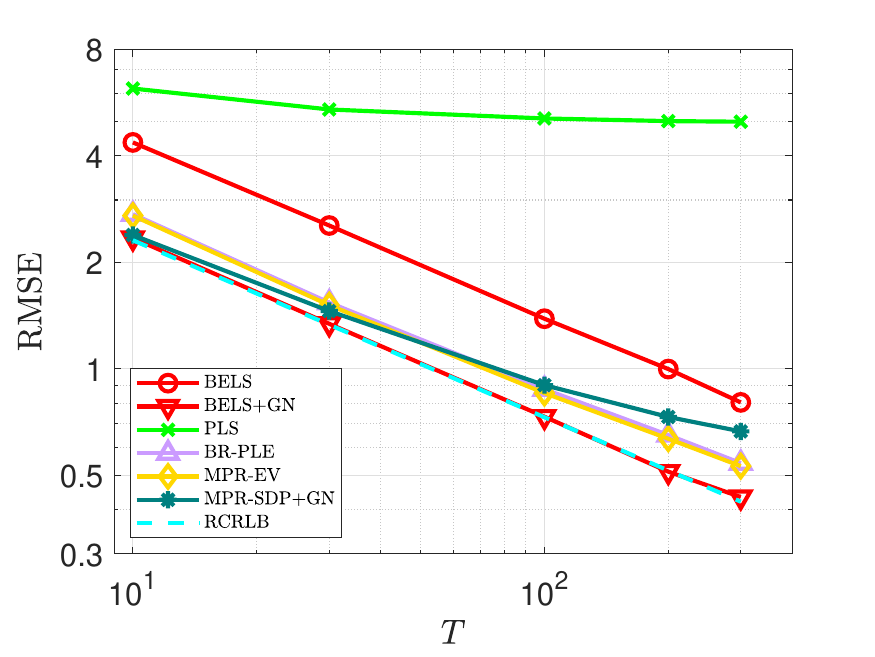}
				\caption{The source is not coplanar with the sensors.}
		\label{fig:3d_noncoplanar_source_rmse}
			\end{subfigure}
			\begin{subfigure}[b]{.24\textwidth}
				\centering
		\includegraphics[width=\textwidth]{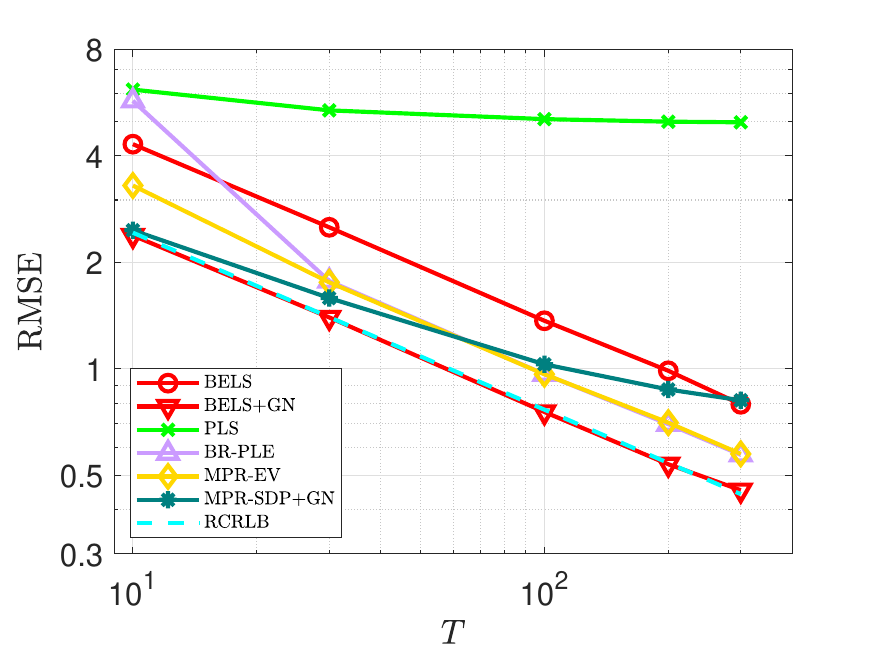}
				\caption{The source and sensors are coplanar.\\~}
				\label{fig:3d_coplanar_source_rmse}
			\end{subfigure}
			\caption{3-D: Biases and RMSEs of the estimators for coplanar sensors.}
			\label{fig:3d_coplanar_rmse}
		\end{figure*}
		
		Figs. \ref{fig:3d_compareVar_varyingT} and \ref{fig:3d_compareVar_varyingsigma} compare the BELS and BELS+GN using true variances of the sine of noises with their estimated counterparts, BELS($\widehat{v}^a_n,\widehat{v}^e_n$)
		and BELS($\widehat{v}^a_n,\widehat{v}^e_n$)+GN. 
		As in the 2-D case, the RMSEs of BELS\((\widehat{v}^a_n, \widehat{v}^e_n)\) and BELS\((\widehat{v}^a_n, \widehat{v}^e_n)\)+GN are virtually indistinguishable from those of BELS and BELS+GN across varying sample sizes and noise levels. This close agreement validates the accuracy and robustness of the proposed estimators for the variances of the sine of both the azimuth and elevation noises, confirming their practical reliability.

Table \ref{table:3d_time} shows the computational time of all estimators with increasing sample sizes based on the avarage of 1000 Monte Carlo runs. Here, the PLS estimator is omitted due to its uncompetitive accuracy.
Our proposed algorithm is the fastest across all sample sizes. The BR-PLE estimator exhibits  $\mathcal{O}(n)$ complexity but incurs slightly higher time costs due to its iterative nature. Both MPR-EV and MPR-SDP+GN exhibit substantially higher complexity: MPR-EV requires $\mathcal{O}(n^3)$ operations dominated by large matrix inversions \cite{Sun2020}, while MPR-SDP+GN involves two CVX calls for SDP-based initialization and  large matrix inversions during GN iterations. Consequently, their execution time grows rapidly with sample size, hindering real-time application in large-scale scenarios. In contrast, our method maintains a significant computational advantage, especially with large sample sizes. 
		
		\begin{table}[ht]
			\centering
			\caption{The average time spent by different algorithms among 1000 experiments. (Unit: seconds)}
			\resizebox{0.7\columnwidth}{!}{
				\begin{tabular}{ccccccc}
					\hline
					\hline
					& $n=100$  & $n=300$  & $n=1000$  & $n=2000$ & $n=3000$   \\ \hline
					\textbf{BELS+GN} & \textbf{0.00129}& \textbf{0.00240}& \textbf{0.01242}& \textbf{0.01648}& \textbf{0.01743}&   \\ \hline
					BR-PLE & 0.02995& 0.09224& 0.28913& 0.57565& 0.86581  \\ \hline
					MPR-EV & 0.35355& 0.77481& 4.06640& 17.79827& 36.49547  \\ \hline
					MPR-SDP+GN & 1.30462& 1.46721& 2.53511& 7.16763& 12.95176  \\ \hline
			\end{tabular}}
			\label{table:3d_time}
		\end{table}
		
		
		\subsection{3-D scenario: fixed and coplanar sensors}
		
	In the 3-D scenario, AOA-based localization only requires noncollinear sensor deployment, unlike TOA/TDOA methods which require noncoplanar arrangements. To verify that coplanar sensor configurations are sufficient, we place 10 sensors in the plane $z=0$ at coordinates: $ p_1 = [0,100,0]^T,$
	$p_2=[0,50,0]^T,$
	$p_3=[50,50,0]^T,$
	$p_4=[50,0,0]^T,$
	$p_5=[50,-50,0]^T,$
  $p_6 = [0,-50,0]^T,$
  $p_7=[0,-100,0]^T,$
  $p_8=[-50,-50,0]^T,$
  $p_9=[-50,0,0]^T,$
  $p_{10}=[-50,50,0]^T$. Two source locations are placed at $p^o=[60,10,10]^T$ (noncoplanar) and $p^o=[60,10,0]^T$ (coplanar). Assumptions \ref{assum_coordinates_3d} and \ref{assum_projection} are satisfied for both cases.

We evaluate the bias and RMSE of all estimators for $T$ $=10,$ $ 30, 100, 200, 300$ with noise standard deviations $\sigma_a = \sigma_e = 0.2$ rad based on the average of 1000 Monte Carlo runs. 
Figs. \ref{fig:3d_noncoplanar_source_bias} and \ref{fig:3d_coplanar_source_bias} reveal that
PLS and MPR-SDP+GN suffer from nonvanishing asymptotic bias, whereas BR-PLE, MPR-EV, BELS, and BELS+GN exhibit diminishing bias, confirming their asymptotic unbiasedness.
Correspondingly, Figs. \ref{fig:3d_noncoplanar_source_rmse} and \ref{fig:3d_coplanar_source_rmse} show that PLS incurs the largest RMSE, while MPR-SDP+GN converges to a nonzero error floor due to its nonvanishing bias. BR-PLE and MPR-EV are consistent (RMSE decreases with $T$) but their RMSEs remain above the RCRLB. In contrast, BELS+GN achieves asymptotic unbiasedness, consistency, and efficiency, with its RMSE converging to the RCRLB in both coplanar and noncoplanar sensor–source arrangements. This result verifies that  noncoplanar sensor deployment is not required for 3-D  AOA-based localization.

		\section{Conclusion}
		\label{sec5}
	In this paper, we have established the asymptotic localizability of the AOA-based localization problem with respect to sensor deployment. Moreover, we have proposed a  consistent and asymptotically efficient two-step estimator for source localization using AOA  measurements under specific conditions related to measurement noise and sensor geometry.
	In the first step, we derive a $\sqrt{n}$-consistent estimator for the true source using the BELS estimator, which involves a necessary procedure to estimate the variance of the sine of the noise when the noise variance is unknown.
		In the second step, we apply one-step GN iteration, using the $\sqrt{n}$-consistent estimator from the first step as the initial value.
		Theoretically, the two-step estimator achieves asymptotic efficiency under appropriate conditions. Notably, the total computational complexity of the two-step estimator is $\mathcal{O}(n)$, making it computationally feasible for practical applications.
		Monte-Carlo simulations also validate its estimation efficiency and low computational complexity  in comparison with other existing estimators.

		\section*{Appendix A: Proofs of results}
		This section contains the proofs of the results not including in the main body of the paper.
		\renewcommand{\thesection}{A}
		\setcounter{equation}{0}
		\renewcommand{\theequation}{A\arabic{equation}}
		\setcounter{lem}{0}
		\setcounter{subsection}{0}
		\renewcommand{\thelem}{A\arabic{lem}}
\subsection{Proof of Theorem 1}
	Under Assumption 2, we have $\mathcal{P} \times \mathcal{P}^o$ is compact and $\left( \arctan \big(\frac{\tilde{y}-y}{\tilde{x}-x}\big)-\big(\frac{\tilde{y}-y^o}{\tilde{x}-x^o}\big)  \right)^2$ is a   continuous and bounded function of $(\tilde{p},p) \in \mathcal{P} \times \mathcal{P}^o$. Then by Lemma \ref{lem_converge_samplemean}, we have $\lim_{n \to \infty}\frac{1}{n}\sum_{i=1}^n(f_i(p)-f_i(p^o))^2 = \mathbb{E}_\mu \Big( \arctan \big(\frac{\tilde{y}-y}{\tilde{x}-x}\big)-\big(\frac{\tilde{y}-y^o}{\tilde{x}-x^o}\big)  \Big)^2$, where $\mathbb{E}_\mu$ is taken over $\tilde{p}=[\tilde{x},\tilde{y}]^T$ with respect to the distribution $\mu$. Suppose that there exists some $p \neq p^o$ such that $\mathbb{E}_\mu \Big( \arctan \big(\frac{\tilde{y}-y}{\tilde{x}-x}\big)-\big(\frac{\tilde{y}-y^o}{\tilde{x}-x^o}\big)  \Big)^2 = 0$. For every such $p$, define $\mathcal{P}_p \eq \left\{ \tilde{p} \in \mathcal{P} \vert  \arctan \big(\frac{\tilde{y}-y}{\tilde{x}-x}\big)=\big(\frac{\tilde{y}-y^o}{\tilde{x}-x^o}\big)\right\}$. Then $\mu(\mathcal{P}_p) =1$. However, $\big(\frac{\tilde{y}-y}{\tilde{x}-x}\big)=\big(\frac{\tilde{y}-y^o}{\tilde{x}-x^o}\big)$ is equivalent to the fact that vectors $\tilde{p}-p$ and $\tilde{p}-p^o$ are parallel. Note that $p \neq p^o$, that is to say that $\tilde{p}$ lies on the line going through $p$ and $p^o$. This contradicts Assumption 3.
		
		This completes the proof.

\subsection{Proof of Lemma 1}
		It is straightforward that
		\begin{align*}
			\label{M_p}
			\frac{1}{n}\sum_{i=1}^n\nabla f_i(p) \nabla f_i(p)^T
			=\frac{1}{n}\sum_{i=1}^n \left( \frac{1}{\|p_i-p\|^4}\begin{bmatrix}
				y_i-y\\-x_i+x
			\end{bmatrix}\left[y_i-y, -x_i+x\right]\right).
		\end{align*}
		Under  Assumptions 2-3, $\frac{1}{\|\tilde{p}-p\|^4}\begin{bmatrix}
				\tilde{y}-y\\-\tilde{x}+x
			\end{bmatrix}\left[\tilde{y}-y, -\tilde{x}+x\right]$ is a continuous and bounded matrix function of $(\tilde{p},p) \in \mathcal{P}\times \mathcal{P}^o$. Then by Lemma \ref{lem_converge_samplemean}, we  obtain that $\frac{1}{n}\sum_{i=1}^n\nabla f_i(p) \nabla f_i(p)^T$ converges uniformly on $\mathcal{P}^o$  as $n \to \infty$. This proves point (i).
		
		Moreover, by the definition of $M^o$, we have
		\begin{align*}
			M^o = \lim_{n\to \infty}\frac{1}{n}\sum_{i=1}^n\nabla f_i(p^o) \nabla f_i(p^o)^T
			=\mathbb{E}_\mu \left( \frac{1}{\|p-p^o\|^4}\begin{bmatrix}
				y-y^o\\-x+x^o
			\end{bmatrix}[y-y^o, -x+x^o]\right),
		\end{align*}
		where $\mathbb{E}_\mu$ is taken over $p = [x,y]^T$ with respect to $\mu$.
		
		Let $\theta$ be a nonzero vector such that
		\begin{align*}
			\theta^T M^o\theta
			=& \mathbb{E}_\mu \left( \frac{1}{\|p-p^o\|^4}\theta^T \begin{bmatrix}
				y-y^o\\-x+x^o
			\end{bmatrix}\left[y-y^o,  -x+x^o\right]\theta\right)\\
			=& \mathbb{E}_\mu \left(\frac{1}{\|p-p^o\|^2}\left[y-y^o,  -x+x^o\right]\theta\right)^2= 0.
		\end{align*}
		Define $\mathcal{P}_\theta \eq \{p\in \mathcal{P}~|~ \left[y-y^o,   -x+x^o\right]\theta = 0\}$ for every such $\theta$.  Then $\mu(\mathcal{P}_\theta)=1$. However, note that $\left[x-x^o,  y-y^o\right]^T = p-p^o$ and $\left[x-x^o,  y-y^o\right]^T$ is perpendicular to $\left[y-y^o,  -x+x^o\right]^T$. Then, for every $\tilde{p} \in \mathcal{P}_\theta$, $\theta$ and $\tilde{p}-p^o$ are parallel. Thus, every $\tilde{p} \in \mathcal{P}_\theta$ lies on the same line, which contradicts Assumption 3. 
		
		This completes the proof.

\subsection{Proof of Proposition 1}
		It is straightforward that 
		\begin{align*}
			\ell_n(p)/n &= -\ln(\sqrt{2\pi}\sigma_a)-\frac{1}{2\sigma_a^2  n}\sum_{i=1}^n \Big(a_i - \arctan\left(\frac{y_i-y}{x_i-x}\right)\Big)^2\\
			&= -\ln(\sqrt{2\pi}\sigma_a) - \frac{1}{2\sigma_a^2 n }\sum_{i=1}^n\big(f_i(p^o)-f_i(p)+\varepsilon_i^a\big)^2 \\ &= -\ln(\sqrt{2\pi}\sigma_a)-\frac{1}{2\sigma_a^2}h_n(p)- \frac{1}{2\sigma_a^2}\Big(\frac{1}{n}\sum_{i=1}^n2\big(f_i(p^o)-f_i(p)\big)\varepsilon_i^a + \frac{1}{n}\sum_{i=1}^n (\varepsilon_i^a)^2\Big).
		\end{align*}
Since both $\mathcal{P}$ and $\mathcal{P}^o$ are bounded under Assumption 2, by Lemma \ref{lem_sqrtn}, it holds that $\frac{1}{n}\sum_{i=1}^n(f_i(p^o)-f_i(p))\varepsilon_i^a \xra{} 0$ almost surely uniformly on $\mathcal{P}^o$. Moreover, $\lim_{n \to \infty}\frac{1}{n}\sum_{i=1}^n (\varepsilon_i^a)^2 = \sigma^2/(100\alpha^2)$  almost surely by Lemma \ref{lem_sqrtn}. As a result,  Theorem 1 yields
		\begin{equation*}
			\frac{1}{n}\ell_n(p) \to -\ln(\sqrt{2\pi}\sigma_a)-\frac{1}{2}-\frac{1}{2\sigma_a^2}h(p) = \ell(p)
		\end{equation*}
		almost surely as $n\ra{}\infty$ uniformly for $p \in \mathcal{P}^o$.
		
		Next, we show that $\nabla^2(-\ell(p^o)) = M^o/\sigma_a^2$. Note that 
		\begin{align*}
			\nabla^2(-\ell_n(p)/n) = \frac{1}{\sigma_a^2 n}\sum_{i=1}^n \big(\nabla f_i(p) \nabla f_i(p)^T - \nabla^2 f_i(p)(a_i-f_i(p))\big).
		\end{align*}
The Hessian matrix
$
   \nabla^2(-\ell_n(p^o)/n)  \to M^o/\sigma_a^2
$ almost surely  at $p=p^o$ as $n\xra{}\infty$ by Lemma 1.
 Moreover, the convergence of $\nabla^2(-\ell_n(p)/n)$ is uniform over $\mathcal{P}^o$ by Lemma \ref{lem_converge_samplemean} as for proving  Lemma 1.
Then, as a direct corollary of the uniform convergence, we have  $\nabla^2(-\ell_n(p^o)/n) \to \nabla^2 (-\ell(p^o))$ as $n \to \infty$, which implies $\nabla^2(-\ell(p^o)) = M^o/\sigma_{\alpha}^2$. 
		
	This completes the proof.

\subsection{Proof of Proposition 2}
		It follows that
		\begin{align*}
			&\frac{1}{n}\left(X^o\right)^TX^o = e^{-\sigma_a^2}\sum_{i=1}^{n}\frac{1}{n}h_i^o(h_i^o)^T\\
			=&~e^{-\sigma_a^2}\sum_{i=1}^{n}\frac{1}{n}\begin{bmatrix}
				\sin^2(a_i^o)	& -\sin(a_i^o)\cos(a_i^o) \\
				-\sin(a_i^o)\cos(a_i^o)	& \cos(a_i^o)
			\end{bmatrix}\\
			=&~e^{-\sigma_a^2}\left(I_2 - \frac{1}{n}\sum_{i=1}^{n}\frac{1}{(r_i^o)^2}(p_i-p^o)(p_i-p^o)^T\right).
		\end{align*}
		Note that the trace of the $2 \times 2$ matrix $\frac{1}{n}\sum_{i=1}^{n}(r_i^o)^{-2}(p_i-p^o)(p_i-p^o)^T$ is 1, then $I_2 - \frac{1}{n}\sum_{i=1}^{n}(r_i^o)^{-2}(p_i-p^o)(p_i-p^o)^T$ is non-singular if $\frac{1}{n}\sum_{i=1}^{n}(r_i^o)^{-2}(p_i-p^o)(p_i-p^o)^T$ is non-singular. Otherwise, suppose that $\frac{1}{n}\sum_{i=1}^{n}(r_i^o)^{-2}(p_i-p^o)(p_i-p^o)^T$ is  singular. Thus there exists some $\theta = [\theta_1, \theta_2]^T \neq 0$ such that for all $i=1,...,n$, $(p_i-p^o)^T\theta = 0$,
		which is equivalent to that for all the sensors $\{p_i\}_{i=1}^n$ lie on a line. This
		contradicts Assumption 3.
		
		Moreover, the existence of $\lim_{n \to \infty} \left(X^o\right)^T X^o / n$ can be established using a method similar to the one used in the proof of Lemma 1.  Further, let $\theta$ be a vector such that $\|\theta\|=1$ and 
		\begin{align*}
			\theta^T\lim_{n \to \infty}\frac{1}{n}\left(X^o\right)^TX^o\theta = e^{-\sigma_a^2}\mathbb{E}_{\mu}\left[1 - \frac{1}{\|p-p^o\|^2}\theta^T(p-p^o)(p-p^o)^T\theta\right]=0.
		\end{align*}
		For every such $\theta$, define $\mathcal{P}_\theta = \{p\in \mathcal{P}~|~ (p-p^o)^T\theta = \|p-p^o\|\}$. Notice that for any $p$, $(p-p^o)^T\theta \leq \|p-p^o\|$. Thus, there holds that $\mu(\mathcal{P}_\theta) = 1$. However, $(p-p^o)^T\theta = \|p-p^o\|$ means that $p-p^o$ and $\theta$ are parallel, which means $\mathcal{P}_\theta$ is a line. This contradicts Assumption 3. 
		
		This completes the proof.

\subsection{Proof of Proposition 3}
		It is straightforward from (11) that 
		\begin{align*}
			\widehat{p}^{\rm UB}_{n} = \left(\left(X^o\right)^TX^o\right)^{-1}\left(X^o\right)^T(X^op^o+W) = p^o + \left(\frac{1}{n}\left(X^o\right)^TX^o\right)^{-1}\left(\frac{1}{n}\left(X^o\right)^TW\right)
		\end{align*}
		and $\mathbb{E}\widehat{p}^{\rm UB}_{n}=p^o$.  
		By \eqref{trigono_2d_3}, we have
		\begin{align*}
		 \sin(\varepsilon_i^a)[\cos(a_i^o),  \sin(a_i^o)]&p^o
		 +\sin(\varepsilon_i^a)r_i^o = \sin(\varepsilon_i^a)[\cos(a_i^o),  \sin(a_i^o)]p_i.
		\end{align*} 
		Thus, the noise sequence $\{w_i\}$ of the model (11) has zero mean and its variance is upper bounded by
		\begingroup
		\allowdisplaybreaks
		\begin{align*}
		&\mathbb{E}\left((\cos(\varepsilon_i^a)-e^{-\sigma_a^2/2})(h_i^o)^T\!p^o \!+\! \sin(\varepsilon_i^a)[\cos(a_i^o), ~ \sin(a_i^o)]p_i\right)^2 \\
			&=\mathbb{E}\left(\big(\cos(\varepsilon_i^a\big)-e^{-\sigma_a^2/2})^2\left((h_i^o)^Tp^o\right)^2\right)
			+ \mathbb{E}\left(\sin^2(\varepsilon_i^a)([\cos(a_i^o),~\sin(a_i^o)]p_i)^2\right)\\
			&\leq\mathbb{V}(\cos(\varepsilon_i^a))\|p^o\|^2 + \mathbb{V}(\sin(\varepsilon_i^a))\|p_i\|^2
		\end{align*}
		\endgroup
		according to Lemma \ref{lem_moment_cos_sin} and Assumption 2. Then, by Lemma \ref{lem_sqrtn}, we have
		\begin{equation}
			\frac{1}{n}\left(X^o\right)^TW = O_p\Big(\frac{1}{\sqrt{n}}\Big).\label{xw}
		\end{equation}
		Thus, we obtain from \eqref{xw} and Proposition 2 that
		\begin{align*}
			\widehat{p}_{n}^{\rm UB} - p^o &= \left(\frac{1}{n}\left(X^o\right)^TX^o\right)^{-1}\left(\frac{1}{n}\left(X^o\right)^TW\right) =O(1)O_p\Big(\frac{1}{\sqrt{n}}\Big)=O_p\Big(\frac{1}{\sqrt{n}}\Big).
		\end{align*}
		This completes the proof.

		\subsection{Proof of Theorem 3}
		\label{proof_thm_be_2d}
		Consider the linear models (6) and (10) and denote
		\begin{align*}
			\Delta X \eq X-X^o  =
			\begin{bmatrix}
				h_1^T-e^{-\sigma_a^2/2} (h^o_1)^T\\
				h_2^T-e^{-\sigma_a^2/2} (h^o_2)^T\\
				\vdots\\
				h_n^T-e^{-\sigma_a^2/2} (h^o_n)^T\\
			\end{bmatrix}
			= \begin{bmatrix}
				(\cos(\varepsilon_1^a) - e^{-\sigma_a^2/2})(h_1^o)^T + \sin(\varepsilon_1^a)[\cos(a_1^o),\sin(a_1^o)]	\\
				(\cos(\varepsilon_2^a) - e^{-\sigma_a^2/2})(h_2^o)^T + \sin(\varepsilon_2^a)[\cos(a_2^o),\sin(a_2^o)]	\\
				\vdots	\\
				(\cos(\varepsilon_n^a) - e^{-\sigma_a^2/2})(h_n^o)^T + \sin(\varepsilon_n^a)[\cos(a_n^o),\sin(a_n^o)]
			\end{bmatrix}.
		\end{align*}
		Then we have
		\begin{align*}
			\frac{1}{n}X^TX = \frac{1}{n}\left(X^o\right)^TX^o + \frac{1}{n}\left(X^o\right)^T\Delta X + \frac{1}{n}(\Delta X)^TX^o + \frac{1}{n}(\Delta X)^T\Delta X.
		\end{align*}
		For the cross term, note that both the sequences $\{\cos(\varepsilon_i^a)-e^{-\sigma_a^2/2}\}_{i=1}^n$ and $\{\sin(\varepsilon_i^a)\}_{i=1}^n$ have zero mean   and finite variance. Thus,  by Lemma \ref{lem_sqrtn} we have
		\begin{align}
				\frac{1}{n}\left(X^o\right)^T\Delta X 
				= e^{-\sigma_a^2/2}\frac{1}{n}\sum_{i=1}^{n}(\cos(\varepsilon_i^a)-e^{-\sigma_a^2/2})h_i^o(h_i^o)^T + e^{-\sigma_a^2/2}\frac{1}{n}\sum_{i=1}^{n}\sin(\varepsilon_i^a)h_i^o[\cos(a_i^o),\sin(a_i^o)] 
				= O_p\Big(\frac{1}{\sqrt{n}}\Big).\label{ecos_esin}
		\end{align}
		using the similar arguments as used in the proof of Proposition 3.
		For the quadratic term, we have
		\begingroup
		\allowdisplaybreaks
		\begin{align*}
			&\frac{1}{n}(\Delta X)^T \Delta X 
			\\=~&\frac{1}{n}\sum_{i=1}^{n}(\cos(\varepsilon_i^a)-e^{-\sigma_a^2/2})^2 h_i^o(h_i^o)^T + \frac{1}{n}\sum_{i=1}^{n}\sin^2(\varepsilon_i^a)\begin{bmatrix}
				\cos^2(a_i^o)	& \cos(a_i^o)\sin(a_i^o) \\
				\cos(a_i^o)\sin(a_i^o)	& \sin^2(a_i^o)
			\end{bmatrix} \\
			&+\frac{1}{n}\sum_{i=1}^{n} \sin(\varepsilon_i^a)(\cos(\varepsilon_i^a)-e^{-\sigma_a^2/2})
			 h_i^o [\cos(a_i^o),  \sin(a_i^o)]+ \frac{1}{n}\sum_{i=1}^{n} \sin(\varepsilon_i^a)(\cos(\varepsilon_i^a)-e^{-\sigma_a^2/2}) \begin{bmatrix}
				\cos(a_i^o)	\\
				\sin(a_i^o)	
			\end{bmatrix}  (h_i^o)^T.
		\end{align*}
		\endgroup
		Assumption 1 indicates that the second moment of $\cos(\varepsilon_i^a)$ and $\sin(\varepsilon_i^a)$ is   available by Lemmas \ref{lem_moment_cos_sin} and \ref{lem_sqrtn}, and hence we obtain
		\begin{equation}
			\label{vcos_vsin}
			\begin{split}
				&\frac{1}{n}(\Delta X)^T \Delta X \\=~&				\frac{1}{n}\sum_{i=1}^{n}\mathbb{E}\big(\sin^2(\varepsilon_i^a)\big)\begin{bmatrix}
					\cos^2(a_i^o)	& \cos(a_i^o)\sin(a_i^o) \\
					\cos(a_i^o)\sin(a_i^o)	& \sin^2(a_i^o)
				\end{bmatrix} + \frac{1}{n}\sum_{i=1}^{n}\mathbb{E}\big((\cos(\varepsilon_i^a) - e^{-\sigma_a^2/2})^2\big)h_i^o(h_i^o)^T + O_p\Big(\frac{1}{\sqrt{n}}\Big)\\
				=~& \frac{1}{n}\sum_{i=1}^{n}\mathbb{V}(\sin(\varepsilon_i^a))\big(I_2-h_i^o(h_i^o)^T\big)+ \frac{1}{n}\sum_{i=1}^{n}\mathbb{V}(\cos(\varepsilon_i^a))h_i^o(h_i^o)^T  + O_p\Big(\frac{1}{\sqrt{n}}\Big)\\
				=~&\big(\mathbb{V}(\cos(\varepsilon_1^a))-\mathbb{V}(\sin(\varepsilon_1^a))\big)
				\frac{1}{n}\sum_{i=1}^{n}h_i^o(h_i^o)^T+\mathbb{V}(\sin(\varepsilon_1^a))I_2 + O_p\Big(\frac{1}{\sqrt{n}}\Big),
			\end{split}
		\end{equation}
		where the second equation holds because for every $i=1,...,n$,
		\begin{equation*}
			\begin{bmatrix}
				\cos^2(a_i^o)	& \cos(a_i^o)\sin(a_i^o) \\
				\cos(a_i^o)\sin(a_i^o)	& \sin^2(a_i^o)
			\end{bmatrix} = I_2 - h_i^o(h_i^o)^T.
		\end{equation*}
		Thus, we have
		\begin{align}
			\frac{1}{n}X^TX = \left(1+e^{\sigma_a^2}\mathbb{V}(\cos(\varepsilon_1^a)) - e^{\sigma_a^2}\mathbb{V}(\sin(\varepsilon_1^a))\right)
			\frac{1}{n}\left(X^o\right)^TX^o + \mathbb{V}(\sin(\varepsilon_1^a))I_2 + O_p\Big(\frac{1}{\sqrt{n}}\Big).\label{xtx}
		\end{align}
		
		Similarly, we have
		\begin{align*}
			\frac{1}{n}X^TY =& \frac{1}{n}\left(X^o\right)^TY + \frac{1}{n}(\Delta X)^TY\\
			=& \frac{1}{n}\left(X^o\right)^TY + \frac{1}{n}(\Delta X)^T W+ \frac{1}{n}(\Delta X)^TX^op^o \\
			=& \frac{1}{n}\left(X^o\right)^TY+ \frac{1}{n}(\Delta X)^T W + O_p\Big(\frac{1}{\sqrt{n}}\Big) .
		\end{align*}
		For the last term, by \eqref{trigono_2d_3} we have $w_i = \big(\cos(\varepsilon_i^a)-e^{-\sigma_a^2/2}\big)(h_i^o)^Tp^o+\sin(\varepsilon_i^a)[\cos(a_i^o),  \sin(a_i^o)]p^o+\sin(\varepsilon_i^a)r_i^o = \big(\cos(\varepsilon_i^a)-e^{-\sigma_a^2/2}\big)(h_i^o)^Tp^o+\sin(\varepsilon_i^a)[\cos(a_i^o),  \sin(a_i^o)]p_i$. Then, using the similar argument as used in \eqref{vcos_vsin}, by Lemmas \ref{lem_moment_cos_sin} and \ref{lem_sqrtn} we derive
		\begingroup
		\allowdisplaybreaks
		\begin{align*}
			&\frac{1}{n}(\Delta X)^TW \\=~& \frac{1}{n}\sum_{i=1}^{n}\big(\cos(\varepsilon_i^a)-e^{-\sigma_a^2/2}\big)^2h_i^o(h_i^o)^Tp^o + \frac{1}{n}\sum_{i=1}^{n}\sin^2(\varepsilon_i^a)\begin{bmatrix}
				\cos^2(a_i^o)	& \cos(a_i^o)\sin(a_i^o) \\
				\cos(a_i^o)\sin(a_i^o)	& \sin^2(a_i^o)
			\end{bmatrix}p_i\\
			&+\frac{1}{n}\sum_{i=1}^{n} \sin(\varepsilon_i^a)\big(\cos(\varepsilon_i^a)-e^{-\sigma_a^2/2}\big)
			\times \left(h_i^o\big[\cos(a_i^o),\sin(a_i^o)\big]p_i + \begin{bmatrix}
				\cos(a_i^o)	\\
				\sin(a_i^o)	
			\end{bmatrix}  (h_i^o)^Tp^o\right)\\
			=&\frac{1}{n}\sum_{i=1}^{n}\mathbb{V}(\cos(\varepsilon_i^a))h_i^o(h_i^o)^Tp^o + \frac{1}{n}\sum_{i=1}^{n}\mathbb{V}(\sin(\varepsilon_i^a))\left(I_2-h_i^o(h_i^o)^T\right)p_i + O_p\Big(\frac{1}{\sqrt{n}}\Big)\\
			=&\big(\mathbb{V}(\cos(\varepsilon_1^a))-\mathbb{V}(\sin(\varepsilon_1^a))\big)
			\frac{1}{n}\sum_{i=1}^{n}h_i^o(h_i^o)^Tp^o+\mathbb{V}(\sin(\varepsilon_1^a))\frac{1}{n}\sum_{i=1}^{n}p_i + O_p\Big(\frac{1}{\sqrt{n}}\Big),
		\end{align*}
		\endgroup
		where the last equation holds by using $(h_i^o)^Tp_i = (h_i^o)^T p^o$, which is obtained by reorganizing  \eqref{trigono_2d_2}. Note that
		\begin{align*}
			\frac{1}{n}\left(X^o\right)^TY &= \frac{1}{n}\left(X^o\right)^TX^op^o + \frac{1}{n}\left(X^o\right)^T W\\&= \frac{1}{n}\left(X^o\right)^TX^op^o + O_p\Big(\frac{1}{\sqrt{n}}\Big).
		\end{align*}
		Then we obtain
		\begin{equation}
		\begin{split}
		 \label{xty}
			\frac{1}{n}X^TY =& \big(1+e^{\sigma_a^2}\mathbb{V}(\cos(\varepsilon_1^a)) - e^{\sigma_a^2}\mathbb{V}(\sin(\varepsilon_1^a))\big)\Big(\frac{1}{n}\left(X^o\right)^TX^op^o\Big) \\&+ \mathbb{V}(\sin(\varepsilon_1^a))\frac{1}{n}\sum_{i=1}^{n}p_i + O_p\Big(\frac{1}{\sqrt{n}}\Big).
				\end{split}
		\end{equation}
		Therefore, there holds that
		\begin{align*}
			\widehat{p}_{n}^{\rm BE} =& \left(\frac{1}{n}X^TX -\mathbb{V}(\sin(\varepsilon_1^a))I_2\right)^{-1}\left(\frac{1}{n}X^TY-\mathbb{V}(\sin(\varepsilon_1^a))\frac{1}{n}\sum_{i=1}^{n}p_i
			\right)\\
			=& \left(\frac{K}{n}\left(X^o\right)^TX^o + O_p\Big(\frac{1}{\sqrt{n}}\Big)\right)^{-1} \left(\frac{K}{n}\left(X^o\right)^TX^op^o + O_p\Big(\frac{1}{\sqrt{n}}\Big)\right)\\
			=&p^o + O_p\Big(\frac{1}{\sqrt{n}}\Big),
		\end{align*}
		where $K = 1+e^{\sigma_a^2}\mathbb{V}(\cos(\varepsilon_1^a)) - e^{\sigma_a^2}\mathbb{V}(\sin(\varepsilon_1^a))=e^{-\sigma_a^2}$ is a non-zero constant by Lemma \ref{lem_moment_cos_sin}.
		Consequently, the BELS estimator $\widehat{p}_{n}^{\rm BE}$ is $\sqrt{n}$-consistent.
		
		This completes the proof.
		
	\subsection{Proof of Theorem 4}
		\label{proof_thm_var_2d}
		By \eqref{xtx} and \eqref{xty}, we have
		\begin{align*}
			Q_n&
			= 	\begin{bmatrix}
				e^{-\sigma_a^2}\Big(\frac{1}{n}\left(X^o\right)^TX^o  \Big)
				&e^{-\sigma_a^2}\frac{1}{n}\left(X^o\right)^TX^op^o    \\
				e^{-\sigma_a^2}(p^o)^T\Big(\frac{1}{n}\left(X^o\right)^TX^o\Big)  & Y^TY/n
			\end{bmatrix}+ 	\begin{bmatrix}
				\mathbb{V}(\sin(\varepsilon_1^a))I_2  
				&  \mathbb{V}(\sin(\varepsilon_1^a))\Big(\frac{1}{n}\sum\limits_{i=1}^{n}p_i\Big)   \\
				\mathbb{V}(\sin(\varepsilon_1^a))\Big(\frac{1}{n}\sum\limits_{i=1}^{n}p_i^T\Big)  & 0
			\end{bmatrix}\\&\hspace{3mm}+ O_p(1/\sqrt{n}).
		\end{align*}
		For $Y^TY/n$, by   \eqref{trigono_2d_2}, Lemmas \ref{lem_moment_cos_sin} and \ref{lem_sqrtn}, we have
		\begin{align*}
			&\frac{1}{n}Y^TY = \frac{1}{n}\Big((p^o)^T\left(X^o\right)^TX^op^o + (p^o)^T\left(X^o\right)^T W + W^TX^op^o + W^T W\Big)\\
			=~&\frac{1}{n}(p^o)^T\left(X^o\right)^TX^op^o + \frac{1}{n} W^T W + O_p\Big(\frac{1}{\sqrt{n}}\Big)\\
			=~& \frac{1}{n}(p^o)^T\left(X^o\right)^TX^op^o + \frac{1}{n}\sum_{i=1}^{n}\mathbb{V}(\cos(\varepsilon_i^a))(p^o)^Th_i^o(h_i^o)^Tp^o +\frac{1}{n}\sum_{i=1}^{n} \mathbb{V}(\sin(\varepsilon_i^a))p_i^T\big(I_2 - h_i^o(h_i^o)^T\big)p_i + O_p\Big(\frac{1}{\sqrt{n}}\Big)\\
			=~&\frac{e^{-\sigma_a^2}}{n}(p^o)^T\left(X^o\right)^TX^op^o + \frac{\mathbb{V}(\sin(\varepsilon_1^a))}{n}\sum_{i=1}^{n}p_i^Tp_i + O_p\Big(\frac{1}{\sqrt{n}}\Big).
		\end{align*}
		Therefore, we have
		\begin{align*}
			Q_n =~& e^{-\sigma_a^2}\frac{1}{n}\begin{bmatrix}
				\left(X^o\right)^TX^o & \left(X^o\right)^TX^op^o  \\
				(p^o)^T\left(X^o\right)^TX^o & (p^o)^T\left(X^o\right)^TX^op^o 
			\end{bmatrix} + \mathbb{V}(\sin(\varepsilon_1^a))S_n + O_p\Big(\frac{1}{\sqrt{n}}\Big).
		\end{align*}
		By Lemmas \ref{lem_converge_samplemean} and \ref{lem_sqrtn}, and the similar arguments used in  the proof of Proposition 2,
the limit of $S_n$ exists and is denoted by 
$S_{\infty} \eq \lim_{n \to \infty}S_n$.
Similarly, the limit  of $$\begin{bmatrix}
				\left(X^o\right)^TX^o & \left(X^o\right)^TX^op^o  \\
				(p^o)^T\left(X^o\right)^TX^o & (p^o)^T\left(X^o\right)^TX^op^o 
			\end{bmatrix} $$ exists, denoted by $U_\infty^o$. Thus, the following limit equation holds
		\begin{equation*}
			Q_{\infty} \eq \lim_{n\to \infty}Q_n= e^{-\sigma_a^2}U_{\infty}^o + \mathbb{V}(\sin(\varepsilon_1^a))S_{\infty},
		\end{equation*}
		and the matrix $U_{\infty}^o$ is positive semi-definite and singular since its columns are linearly dependent. And $Q_{\infty}$ is non-singular due to the non-singularity of $S_{\infty}$.  Then, by Lemma \ref{lem_eigen}, we have
		\begin{equation*}
			\lambda_{\rm max}(Q_{\infty}^{-1}\mathbb{V}(\sin(\varepsilon_1^a))S_\infty) = 1.
		\end{equation*}
		Note that $Q_n-Q_\infty = O_p(1/\sqrt{n})$ and $S_n-S_\infty = O_p(1/\sqrt{n})$, and $\lambda_{\rm max}(Q_n^{-1}S_n)$ is a continuous function of $Q_n$ and $S_n$.
		Thus, we have
$
			\lambda_{\rm max}(Q_n^{-1}S_n)-\lambda_{\rm max}(Q_\infty^{-1}S_\infty) = O_p(1/\sqrt{n})
$, and further
		\begin{align*}
			\widehat{v}^a_n &= \frac{1}{\lambda_{\rm max}(Q_n^{-1}S_n)} = \frac{1}{\lambda_{\rm max}(Q_\infty^{-1}S_\infty)} + O_p(1/\sqrt{n})\\ &= \mathbb{V}(\sin(\varepsilon_1^a)) + O_p(1/\sqrt{n}).
		\end{align*}
		This completes the proof.
		
 	\subsection{Proof of Theorem 5}
		\label{proof_lem_3d}
	Similar to the argument used in proving Theorem 1, it is straightforward to derive  that
		\begin{equation}
			\label{expect_ml_3d}
			\begin{split}
				&\lim_{n \to \infty}\frac{1}{n}\sum_{i=1}^n \Big(f_{i}(p)-f_{i}(p^o)\Big)^T\Big(f_{i}(p)-f_{i}(p^o)\Big)  \\
				=& \sigma_a^{-2}\mathbb{E}_\mu   \left(\arctan\left(\frac{\tilde{y}-y}{\tilde{x}-x}\right)-\arctan\left(\frac{\tilde{y}-y^o}{\tilde{x}-x^o}\right)\right)^2  \\
				&+\sigma_e^{-2}\mathbb{E}_\mu   \Bigg(\arctan\left(\frac{\tilde{z}-z}{\sqrt{(\tilde{x}-x)^2+(\tilde{y}-y)^2}}\right)-\arctan\left(\frac{\tilde{z}-z^o}{\sqrt{(\tilde{x}-x^o)^2+(\tilde{y}-y^o)^2}}\right)\Bigg)^2  ,
			\end{split}
		\end{equation}
		where $\mathbb{E}_\mu$ is taken over $\tilde{p}=[\tilde{x},\tilde{y},\tilde{z}]^T$ with respect to $\mu$. We have the conclusion that
		\begin{equation}
			\label{E_arctanyx}
			\mathbb{E}_\mu   \left(\arctan\left(\frac{\tilde{y}-y}{\tilde{x}-x}\right)-\arctan\left(\frac{\tilde{y}-y^o}{\tilde{x}-x^o}\right)\right)^2 
		\end{equation}
		reaches its unique minimum at $p_{1:2}=p^o_{1:2}$ by Theorem 1 for the 2-D scenario, where $p_{1:2}$ and $p^o_{1:2}$ represent the first two coordinates of $p$ and $p^o$, respectively.  
		Therefore,  it suffices to verify that
		\begin{align*}
			\mathbb{E}_\mu \Bigg(&\arctan\left(\frac{\tilde{z}-z}{\sqrt{(\tilde{x}-x)^2+(\tilde{y}-y)^2}}\right)-\arctan\left(\frac{\tilde{z}-z^o}{\sqrt{(\tilde{x}-x^o)^2+(\tilde{y}-y^o)^2}}\right)\Bigg)^2 
		\end{align*}
		is uniquely minimized at $p= p^o$. 
		Suppose there is some $p \neq p^o$ so that $p_{1:2}=p^o_{1:2}$ and $\mathbb{E}_\mu \Bigg(\arctan\left(\frac{\tilde{z}-z}{\sqrt{(\tilde{x}-x)^2+(\tilde{y}-y)^2}}\right)-\arctan\left(\frac{\tilde{z}-z^o}{\sqrt{(\tilde{x}-x^o)^2+(\tilde{y}-y^o)^2}}\right)\Bigg)^2=0$. Define the set $\mathcal{P}'_p \eq \bigg\{\tilde{p}\in \mathcal{P}~\Bigg|~\arctan\Big(\frac{\tilde{z}-z}{\sqrt{(\tilde{x}-x)^2+(\tilde{y}-y)^2}}\Big)=\arctan\Big(\frac{\tilde{z}-z^o}{\sqrt{(\tilde{x}-x^o)^2+(\tilde{y}-y^o)^2}}\Big), p_{1:2}=p^o_{1:2}\bigg\}$ for every such $p$. Then $\mu(\mathcal{P}'_p)=1$.
		Note that
		$
		[\tilde{x}-x,  \tilde{y}-y]^T = \tilde{p}_{1:2}-p_{1:2}=\tilde{p}_{1:2}-p^o_{1:2} = [\tilde{x}-x^o,  \tilde{y}-y^o]^T.
		$
		Thus, $\arctan\left(\frac{\tilde{z}-z}{\sqrt{(\tilde{x}-x)^2+(\tilde{y}-y)^2}}\right)=\arctan\left(\frac{\tilde{z}-z^o}{\sqrt{(\tilde{x}-x^o)^2+(\tilde{y}-y^o)^2}}\right)$ is equivalent to $\tilde{z}-z = \tilde{z}-z^o$. This derives that $z=z^o$ and further $p=p^o$.
		
		\subsection{Proof of Lemma 3}
		The convergence of $\frac{1}{n}\sum_{i=1}^n\big(\nabla f_{i}(p)\big)^T \nabla f_{i}(p)$ can be derived using arguments similar to those used in the proof of Lemma 1 for the 2-D scenario. 
		Further, define $\rho \eq [\tilde{y}-y^o, -\tilde{x}+x^o,0]^T$ and $\psi \eq [(\tilde{x}-x^o)(\tilde{z}-z^o),(\tilde{y}-y^o)(\tilde{z}-z^o),-(\tilde{x}-x^o)^2-(\tilde{y}-y^o)^2]^T$. Then, by definition,
		\begin{align*}
			M^o
			= \mathbb{E}_\mu \Bigg( \frac{\sigma_a^{-2}}{\|\tilde{p}_{1:2}-p^o_{1:2}\|^4}\rho \rho^T  + \frac{\sigma_e^{-2}}{\|\tilde{p}-p^o\|^4\|\tilde{p}_{1:2}-p^o_{1:2}\|^2} \psi \psi^T\Bigg),
		\end{align*}
		where $\mathbb{E}_\mu$ is taken over $\tilde{p}=[\tilde{x},\tilde{y},\tilde{z}]^T$ with respect to $\mu$. Suppose there exists some $\theta = [\theta_1, \theta_2,\theta_3]^T$ such that $\theta \neq 0 $ and $\theta^T \lim_{n \to \infty}\frac{1}{n}\big(\nabla f_{i}(p^o)\big)^T \nabla f_{i}(p^o) \theta = 0$. Then there holds that 
		\begin{align}
			\label{M_3d_quadra_1}
			&\mathbb{E}_\mu \left( \theta^T \rho \rho^T\theta \right)=0,\\ \label{M_3d_quadra_2}
			&\mathbb{E}_\mu\left(\theta^T \psi\psi^T\theta \right) = 0.
		\end{align}
		
		By Lemma 1(ii), it is straightforward to derive that \eqref{M_3d_quadra_1} is equivalent to $\theta_1=\theta_2=0$. Thus, \eqref{M_3d_quadra_2} is equivalent to
		\begin{equation*}
			\mathbb{E}_\mu\left[\big((\tilde{x}-x^o)^2+(\tilde{y}-y^o)^2\big)(\theta_3)^2 \right] = 0.
		\end{equation*}
		Thus, we obtain $\theta_3=0$ as well, which contradicts the assumption $\theta \neq 0$. Therefore, $M^o$ is non-singular.
		
		This completes the proof.

\subsection{Proof of Proposition 5}

			By definition, we have
			\begin{align*}
				&\frac{1}{n}\left(\Phi^o\right)^T\Phi^o = e^{-\sigma_e^2}\sum_{i=1}^{n}\frac{1}{n}\cos^2(e_i^o)=e^{-\sigma_e^2}\frac{1}{n}\sum_{i=1}^{n}\frac{(x_i-x^o)^2 + (y_i-y^o)^2}{(x_i-x^o)^2 + (y_i-y^o)^2 + (z_i-z^o)^2}.
			\end{align*}
			On the one hand, $(x_i-x^o)^2 + (y_i-y^o)^2 + (z_i-z^o)^2$ is upper bounded by a positive constant since $\mathcal{P}$ and $\mathcal{P}^o$ are bounded. On the other hand, since $\mathcal{P}_{1:2} \cap \mathcal{P}^o_{1:2} = \emptyset$ and both sets are bounded, $(x_i-x^o)^2 + (y_i-y^o)^2$ is  bounded from below by a positive constant. Therefore, $(\Phi^o)^T\Phi^o/n$ is uniformly bounded from below  by a positive constant regardless of $n$.
			
			This completes the proof.
		
		\subsection{Proof of Proposition 6}
		\label{proof_thm_ub_z}
		It follows from (26) that
		\begin{align*}
			\widehat{z}_n^{\rm UB} &= \left((\Phi^o)^T\Phi^o\right)^{-1}(\Phi^o)^T(\Phi^o z^o + \eta) = z^o + \left(\frac{1}{n}(\Phi^o)^T\Phi^o\right)^{-1}\Big(\frac{1}{n}(\Phi^o)^T\eta\Big).
		\end{align*}
		Under Assumption 5, each element of the  noise  sequence of the model (26) is of zero mean and its variance
		\begin{align*}
			&\mathbb{E}\Big(\big(e^{-\sigma_e^2/2}-\cos(\varepsilon_i^e)\big)\cos(e_i^o)z^o+\sin(\varepsilon_i^e)\big(d_i^o+\sin(e_i^o)z^o\big)\Big)^2\\
			&=\mathbb{E}\Big(\big(\cos(\varepsilon_i^e)-e^{-\sigma_e^2/2}\big)^2\cos^2(e_i^0)(z^o)^2\Big) + \mathbb{E}\Big(\sin^2(\varepsilon_i^e)\big(d_i^o+\sin(e_i^o)z^o\big)^2\Big)\\
			&\leq \mathbb{V}(\cos(\varepsilon_i^e))(z^o)^2 + 2\mathbb{V}(\sin(\varepsilon_i^e))\big((d_i^o)^2 + (z^o)^2\big),
		\end{align*}
		is uniformly bounded by Lemma \ref{lem_moment_cos_sin}. 
		Thus, by Lemma \ref{lem_sqrtn}, we have
		\begin{equation}
			\label{phiomega}
			\frac{1}{n}(\Phi^o)^T\eta = O_p(1/\sqrt{n}).
		\end{equation}
		Additionally, combining with Proposition 5 one derives that
		\begin{equation*}
			\widehat{z}_n^{\rm UB} - z^o  = O(1)O_p(1/\sqrt{n}) = O_p(1/\sqrt{n}).
		\end{equation*}
		This completes the proof.
		
		\subsection{Proof of Theorem 7}
		\label{proof_BELS_3D_case}
		It follows from (23) and (26) that
		\begin{align*}
			\Delta\Phi &\eq \Phi - \Phi^o 
			=\begin{bmatrix}
				-\cos(e_1)+e^{-\sigma_e^2/2}\cos(e_1^o)\\
				-\cos(e_2)+e^{-\sigma_e^2/2}\cos(e_2^o)\\
				\vdots\\
				-\cos(e_n)+e^{-\sigma_e^2/2}\cos(e_n^o)
			\end{bmatrix}
			= \begin{bmatrix}
				-(\cos(\varepsilon_1^e)-e^{-\sigma_e^2/2})\cos(e_1^o) + \sin(\varepsilon_1^e)\sin(e_1^o)	\\
				-(\cos(\varepsilon_2^e)-e^{-\sigma_e^2/2})\cos(e_2^o) + \sin(\varepsilon_2^e)\sin(e_2^o)	\\
				\vdots	\\
				-(\cos(\varepsilon_n^e)-e^{-\sigma_e^2/2})\cos(e_n^o) + \sin(\varepsilon_n^e)\sin(e_n^o)
			\end{bmatrix}.
		\end{align*}
		Then we have
		\begin{equation*}
			\frac{1}{n}\Phi^T\Phi = \frac{1}{n}(\Phi^o)^T\Phi^o + \frac{2}{n}(\Phi^o)^T\Delta\Phi + \frac{1}{n}\Delta\Phi^T\Delta\Phi.
		\end{equation*}
		For the cross term, by the similar argument as \eqref{ecos_esin}, we have
		\begin{align*}
			\frac{2}{n}(\Phi^o)^T\Delta\Phi 
			=~&e^{-\sigma_e^2/2}\frac{2}{n}\sum_{i=1}^{n}\big(\cos(\varepsilon_i^e)-e^{-\sigma_e^2/2}\big)\cos^2(e_i^o) - e^{-\sigma_e^2/2}\frac{1}{n}\sum_{i=1}^{n}\sin(\varepsilon_i^e)\sin(e_i^o)\cos(e_i^o)\\
			=~&O_p(1/\sqrt{n}).
		\end{align*}
		For the quadratic term, we have
		\begingroup
		\allowdisplaybreaks
		\begin{align*}
			&\frac{1}{n}\Delta\Phi^T\Delta\Phi\\=~& \frac{1}{n}\sum_{i=1}^{n}(\cos(\varepsilon_i^e)-e^{-\sigma_e^2/2})^2\cos^2(e_i^o)	+\frac{1}{n}\sum_{i=1}^{n} \sin^2(\varepsilon_i^e)\sin^2(e_i^o) - \frac{2}{n}\sum_{i=1}^{n}\sin(\varepsilon_i^e)(\cos(\varepsilon_i^e)-e^{-\sigma_e^2/2})\sin(e_i^o)\cos(e_i^o).
		\end{align*}
		By the similar argument as \eqref{vcos_vsin}, we derive
		\begin{align*}
			&\frac{1}{n}\Delta\Phi^T\Delta\Phi\\ =~& \frac{1}{n}\sum_{i=1}^{n}\mathbb{E}\big((\cos(\varepsilon_i^e)-e^{-\sigma_e^2/2})^2\cos^2(e_i^o)\big) + \frac{1}{n}\sum_{i=1}^{n} \mathbb{E}\left(\sin^2(\varepsilon_i^e)\right)\sin^2(e_i^0) + O_p(1/\sqrt{n})\\
			=~& \frac{1}{n}\sum_{i=1}^{n}\mathbb{V}(\cos(\varepsilon_i^e))\cos^2(e_i^o) + \frac{1}{n}\sum_{i=1}^{n}\mathbb{V}(\sin(\varepsilon_i^e))(1-\cos^2(e_i^o)) + O_p(1/\sqrt{n})\\
			=~& e^{\sigma_e^2}\big(\mathbb{V}(\cos(\varepsilon_i^e))-\mathbb{V}(\sin(\varepsilon_i^e))\big)\Big(\frac{1}{n}(\Phi^o)^T\Phi^o \Big)+ \mathbb{V}(\sin(\varepsilon_1^e)) + O_p(1/\sqrt{n}).
		\end{align*}
		\endgroup
		Thus, there holds that
		\begin{align}
			\frac{1}{n}\Phi^T\Phi &= \left(1+e^{\sigma_e^2}\mathbb{V}(\cos(\varepsilon_i^e)) - e^{\sigma_e^2}\mathbb{V}(\sin(\varepsilon_i^e))\right)\Big(\frac{1}{n}(\Phi^o)^T\Phi^o \Big) + \mathbb{V}(\sin(\varepsilon_1^e)) + O_p(1/\sqrt{n}).\label{hatphihatphi}
		\end{align}
		
		Additionally, $\Phi^T\widehat{\Gamma}/n$ can be decomposed as
		\begin{align*}
			\frac{1}{n}\Phi^T\widehat{\Gamma} &= \frac{1}{n}\Phi^T\Gamma + \frac{1}{n}\Phi^T(\widehat{\Gamma}-\Gamma)= \frac{1}{n}\Phi^T\Gamma + \frac{1}{n}\sum_{i=1}^{n}\cos(e_i)(\widehat{r}_i-r_i^o)\sin(e_i).
		\end{align*}
		Since $\widehat{r_i}-r_i^o = O_p(1/\sqrt{n})$ for all $i=1,...,n$ by (29), we have $\Phi^T\widehat{\Gamma}/n-\Phi^T\Gamma/n = O_p(1/\sqrt{n})$.
		By the similar argument as used for deriving \eqref{ecos_esin} and \eqref{vcos_vsin}, we have
		\begin{align*}
			\frac{1}{n}\Phi^T\Gamma =& \frac{1}{n}(\Phi^o)^T\Gamma + \frac{1}{n}\Delta\Phi^T\Gamma
			=\frac{1}{n}(\Phi^o)^T\Gamma + \frac{1}{n}\Delta\Phi^T\Phi^o z^o+ \frac{1}{n}\Delta\Phi^T\eta
			=\frac{1}{n}(\Phi^o)^T\Gamma+ \frac{1}{n}\Delta\Phi^T\eta + O_p(1/\sqrt{n}),
		\end{align*}
		where
		\begin{align*}
			&\frac{1}{n}\Delta\Phi^T\eta \\ =& \frac{1}{n}\sum_{i=1}^{n}\big(\cos(\varepsilon_i^e)- e^{-\sigma_e^2/2}\big)^2\cos^2(e_i^o)z^o + \frac{1}{n}\sum_{i=1}^{n}\sin^2(\varepsilon_i^e)\big(\sin(e_i^o)d_i^o+\sin^2(e_i^o)z^o\big)\\
			&-\frac{1}{n}\sum_{i=1}^{n}\sin(\varepsilon_i^e)\big(\cos(\varepsilon_i^e)-e^{-\sigma_e^2/2}\big)  \big(\sin(e_i^o)\cos(e_i^o)z^o+\cos(e_i^o)(d_i^o+\sin(e_i^o)z^o)\big)\\
			=&\frac{1}{n}\sum_{i=1}^{n}\mathbb{V}(\cos(\varepsilon_i^e))\cos^2(e_i^o)z^o + \frac{1}{n}\sum_{i=1}^{n}\mathbb{V}(\sin(\varepsilon_i^e))(z_i-\cos^2(e_i^o)z^o) + O_p(1/\sqrt{n})\\
			=&e^{\sigma_e^2}\left(\mathbb{V}(\cos(\varepsilon_1^e))-\mathbb{V}(\sin(\varepsilon_1^e)\right)\Big(\frac{1}{n}(\Phi^o)^T\Phi^o \Big)z^o + \mathbb{V}(\sin(\varepsilon_1^e))\Big(\frac{1}{n}\sum_{i=1}^{n}z_i\Big) + O_p(1/\sqrt{n}).
		\end{align*}
		Note that
		\begin{align*}
			\frac{1}{n}(\Phi^o)^T\Gamma = \frac{1}{n}(\Phi^o)^T\Phi^o z^o + \frac{1}{n}(\Phi^o)^T\eta= \frac{1}{n}(\Phi^o)^T\Phi^o z^o + O_p(1/\sqrt{n}).
		\end{align*}
		Then it holds that
		\begin{align}
			\frac{1}{n}\Phi^T\widehat{\Gamma}  = \left(1+e^{\sigma_e^2}\mathbb{V}(\cos(\varepsilon_1^a)) - e^{\sigma_e^2}\mathbb{V}(\sin(\varepsilon_1^e))\right)
			\Big(\frac{1}{n}(\Phi^o)^T\Phi^o z^o\Big)
			+\mathbb{V}(\sin(\varepsilon_1^e))\Big(\frac{1}{n}\sum_{i=1}^{n}z_i\Big) + O_p(1/\sqrt{n}).\label{hatphihatbeta}
		\end{align}
		Therefore, combining \eqref{hatphihatphi} with \eqref{hatphihatbeta} derives that
		\begin{align*}
			&\widehat{z}_n^{\rm BE} =  \left(\frac{1}{n}\Phi^T\Phi - \mathbb{V}(\sin(\varepsilon_1^e))\right)^{-1}\!\!
			\left(\frac{1}{n}\Phi^T\widehat{\Gamma} - \mathbb{V}(\sin(\varepsilon_1^e))\frac{1}{n}\sum_{i=1}^{n}z_i\right)\\
			& =\left(\frac{K'}{n}(\Phi^o)^T\Phi^o + O_p\Big(\frac{1}{\sqrt{n}}\Big)\right)^{-1}\!\!
			\left(\frac{K'}{n} (\Phi^o)^T\Phi^o z^o + O_p\Big(\frac{1}{\sqrt{n}}\Big)\right)\\
		& =z^o + O_p(1/\sqrt{n}),
		\end{align*}
		where $K' = 1+e^{\sigma_e^2}\mathbb{V}(\cos(\varepsilon_1^e)) - e^{\sigma_e^2}\mathbb{V}(\sin(\varepsilon_1^e))=e^{-\sigma_e^2}$ is a non-zero constant by Lemma \ref{lem_moment_cos_sin}.
		Consequently, the BELS estimator $\widehat{z}_n^{\rm BE}$ is $\sqrt{n}$-consistent.
		
		This completes the proof.
		
		\subsection{Proof of Theorem 9}
		\label{proof_thm_var_3d}
		It follows from \eqref{hatphihatphi} and \eqref{hatphihatbeta} that
		\begin{align*}
			R_n
			&=   \begin{bmatrix}
				\Phi^T\Phi/n & \Phi^T\widehat{\Gamma}/n\\
				\widehat{\Gamma}^T\Phi/n & \widehat{\Gamma}^T\widehat{\Gamma}/n
			\end{bmatrix} \\ &= \begin{bmatrix}
				e^{-\sigma_e^2}(\Phi^o)^T\Phi^o/n & e^{-\sigma_e^2}(\Phi^o)^T\Phi^o z^o/n\\
				e^{-\sigma_e^2}(\Phi^o)^T\Phi^o z^o/n &\widehat{\Gamma}^T\widehat{\Gamma}/n
			\end{bmatrix}+\begin{bmatrix}
				\mathbb{V}(\sin(\varepsilon_1^e)) & \mathbb{V}(\sin(\varepsilon_1^e))\bar{z}\\
				\mathbb{V}(\sin(\varepsilon_1^e))\bar{z} &0
			\end{bmatrix}+ O_p(1/\sqrt{n}).
		\end{align*}
		
		The term $\widehat{\Gamma}^T\widehat{\Gamma}/n$ can be decomposed as
		\begin{align*}
			\frac{1}{n}\widehat{\Gamma}^T\widehat{\Gamma} = \frac{1}{n}(\Gamma+\widehat{\Gamma}-\Gamma)^T(\Gamma+\widehat{\Gamma}-\Gamma)=\frac{1}{n}\Gamma^T\Gamma + \frac{2}{n}\Gamma^T(\widehat{\Gamma}-\Gamma) + \frac{1}{n}\|\widehat{\Gamma}-\Gamma\|^2.
		\end{align*}
		For the last addend, we have
		\begin{equation*}
			\frac{1}{n}\|\widehat{\Gamma}-\Gamma\|^2 =  \frac{1}{n}\sum_{i=1}^n\sin^2(e_i)(\widehat{r_i}-r_i^o)^2 \leq \frac{1}{n}\sum_{i=1}^n(\widehat{r_i}-r_i^o)^2 = O_p\Big(\frac1n\Big)
		\end{equation*}
		following from $\widehat{r_i}-r_i^o = O_p(1/\sqrt{n})$.
		
		For the term $2\Gamma^T(\widehat{\Gamma}-\Gamma)/n$, we have
		\begin{align*}
			\frac{2}{n}\Gamma^T(\widehat{\Gamma}-\Gamma) = \frac{2}{n}\sum_{i=1}^n\big(\sin(e_i)r_i+\cos(e_i)z_i\big)\sin(e_i)(\widehat{r_i}-r_i^o)= O_p(1/\sqrt{n})
		\end{align*} 
		by the uniform boundedness of the sequence $\{(\sin(e_i)r_i^o+\cos(e_i)z_i)\sin(e_i)\}_{i=1}^n$
		under Assumption 5.
		
		For the term $\Gamma^T\Gamma/n$, there holds that
		\begin{equation*}
			\frac{1}{n}\Gamma^T\Gamma = \frac{1}{n}(\Phi^o)^T\Phi^o (z^o)^2 + \frac{1}{n} z^o(\Phi^o)^T\eta + \frac{1}{n}\eta^T\eta.
		\end{equation*}
		By \eqref{phiomega}, we have
		\begin{equation*}
			\frac{2}{n}z^o (\Phi^o)^T\eta = O_p(1/\sqrt{n}).
		\end{equation*}
		For the addend $\eta^T\eta/n$, we have
		\begin{align*}
			&\frac{1}{n}\eta^T\eta\\ =& \frac{1}{n}\sum_{i=1}^n\Big[-\big(\cos(\varepsilon_i^e)-e^{-\sigma_e^2/2}\big)\cos(e_i^o)z^o+\sin(\varepsilon_i^e)(d_i^o+\sin(e_i^o)z^o)\Big]^2\\
			=& \frac{1}{n}\sum_{i=1}^n\Big[(\cos(\varepsilon_i^e)-e^{-\sigma_e^2/2})^2\cos^2(e_i^o)(z^o)^2 - 2\big(\cos(\varepsilon_i^e)\!-\!e^{-\sigma_e^2/2}\big)\cos(e_i^o)z^o\sin(\varepsilon_i^e)\big(d_i^o+\sin(e_i^o)z^o\big)\\
			&~~~~~~~~~+ \sin^2(\varepsilon_i^e)\big(d_i^o+\sin(e_i^o)z^o\big)^2
			\Big].
		\end{align*}
		Note that
		\begin{align*}
			&~~~~(d_i^o + \sin(e_i^o)z^o)^2\\& = (d_i^o)^2 + 2d_i^o\sin(e_i^o)z^o + \sin^2(e_i^o)(z^o)^2\\
			&= (r_i^o)^2+(z_i-z^o)^2 + 2(z_i-z^o) z^o+ (z^o)^2-\cos^2(e_i^o)(z^o)^2\\
			& = (r_i^o)^2 + z_i^2 + (z^o)^2 - 2z_iz^o - (z^o)^2 + 2z_iz^o-\cos^2(e_i^o)(z^o)^2\\
			& = \widehat{r_i}^2 +z_i^2 -\cos^2(e_i^o)(z^o)^2 + O_p(1/n).
		\end{align*}
		Under Assumption 5, by Lemma \ref{lem_moment_cos_sin}  we derive
		\begin{align*}
			&\frac{1}{n}\sum_{i=1}^n (\cos(\varepsilon_i^e)-e^{-\sigma_e^2/2})^2\cos^2(e_i^o)(z^o)^2 = e^{\sigma_e^2}\mathbb{V}(\cos(\varepsilon_1^e))\Big(\frac{1}{n}(\Phi^o)^T\Phi^o\Big) (z^o)^2 + O_p(1/\sqrt{n}),\\
			&\frac{1}{n}\sum_{i=1}^n \sin^2(\varepsilon_i^e)(d_i^o+\sin(e_i^0)z^o)^2 = \mathbb{V}(\sin(\varepsilon_i^e))\frac{1}{n}\sum_{i=1}^n \left(\widehat{r_i}^2+z_i^2\right) - e^{\sigma_e^2}\mathbb{V}(\sin(\varepsilon_i^e))\Big(\frac{1}{n}(\Phi^o)^T\Phi^o \Big)(z^o)^2 + O_p(1/\sqrt{n}),\\
			&\frac{1}{n}\sum_{i=1}^n 2(\cos(\varepsilon_i^e)-e^{-\sigma_e^2/2})\cos(e_i^o)z^o\sin(\varepsilon_i^e)(d_i^o-\sin(e_i^o)z^o) = O_p(1/\sqrt{n}).
		\end{align*}
		Thus, we have
		\begin{align*}
			\frac{1}{n}\widehat{\Gamma}^T\widehat{\Gamma} &= e^{-\sigma_e^2}\frac{1}{n}(\Phi^o)^T\Phi^o (z^o)^2 \!+\! \mathbb{V}(\sin(\varepsilon_i^e))\frac{1}{n}\sum_{i=1}^n \big(\widehat{r_i}^2\!+\!z_i^2\big) + O_p(1/\sqrt{n}).
		\end{align*}
		Therefore, it holds that
		\begin{align*}
			R_n = e^{-\sigma_e^2}\frac{1}{n}\begin{bmatrix}
				(\Phi^o)^T\Phi^o & (\Phi^o)^T\Phi^o z^o\\
				(\Phi^o)^T\Phi^o z^o & (\Phi^o)^T\Phi^o (z^o)^2
			\end{bmatrix}
			+ \mathbb{V}(\sin(\varepsilon_1^e))U_n + O_p(1/\sqrt{n}).
		\end{align*}
		
		By the similar arguments as used in the proof of Theorem 4, we derive that
		\begin{align*}
			R_\infty = e^{-\sigma_e^2}\lim_{n\to \infty}\frac{1}{n}\begin{bmatrix}
				(\Phi^o)^T\Phi^o & (\Phi^o)^T\Phi^o z^o\\
				(\Phi^o)^T\Phi^o z^o & (\Phi^o)^T\Phi^o (z^o)^2
			\end{bmatrix} + \mathbb{V}(\sin(\varepsilon_1^e))U_\infty.
		\end{align*}
		Again, by the similar argument in the proof of Theorem 4, we obtain that
		\begin{equation*}
			\lim_{n\to \infty} \frac{1}{n}\begin{bmatrix}
				(\Phi^o)^T\Phi^o & (\Phi^o)^T\Phi^o z^o\\
				(\Phi^o)^T\Phi^o z^o & (\Phi^o)^T\Phi^o (z^o)^2
			\end{bmatrix}
		\end{equation*}
		is positive semi-definite and further is
		singular due to the linearly dependent  columns, and $U_{\infty}$ is non-singular. By Lemma \ref{lem_eigen}, we have
		\begin{equation*}
			\lambda_{\rm max}(R_\infty^{-1}\mathbb{V}(\sin(\varepsilon_1^e))U_{\infty})=1.
		\end{equation*}
		Then, we derive that
		\begin{align*}
			\widehat{v}_n^e = \frac{1}{\lambda_{\rm max}(R_n^{-1}U_n)} = \frac{1}{\lambda_{\rm max}(R_\infty^{-1}U_{\infty})} + O_p(1/\sqrt{n}) = \mathbb{V}(\sin(\varepsilon_1^e)) + O_p(1/\sqrt{n}).
		\end{align*}
		This completes the proof.

		\section*{Appendix B: Auxiliary results}
		\renewcommand{\thesection}{B}
		\setcounter{equation}{0}
		\renewcommand{\theequation}{B\arabic{equation}}
		\setcounter{lem}{0}
		\setcounter{subsection}{0}
		\renewcommand{\thelem}{B\arabic{lem}}
		
		This subsection contains the auxiliary lemmas used for the proofs.
		
		\begin{lem} 
			\label{lem_trigono} 
			For convenience,  we list the identities associated with the AOA localization.
			\begin{enumerate}[(i)]
				\item For the 2-D scenario, following the notations in Section II, there holds that for every $i=1,..,n$,
				\begin{subequations}
					\label{trigono_2d}
					\begin{align}
						& (x_i-x^o)\sin(a_i^o)-(y_i-y^o)\cos(a_i^o) = 0, \label{trigono_2d_2}\\
						& (x_i-x^o)\cos(a_i^o) + (y_i-y^o)\sin(a_i^o) = r_i^o.\label{trigono_2d_3}
					\end{align}
				\end{subequations}
				\item For the 3-D scenario, following the notations in Section III, it holds that for every $i=1,..,n$,
				\begin{subequations}
					\label{trigono_3d}
					\begin{align}
				&r_i^o\sin(e_i^o)-(z_i-z^o)\cos(e_i^o) = 0,\\
				&r_i^o\cos(e_i^o)+(z_i-z^o)\sin(e_i^o)=d_i^o.
					\end{align}
				\end{subequations}
			\end{enumerate}
		\end{lem}
		\begin{proof}
			The proof is straightforward by checking the geometry.
		\end{proof}
		
		\begin{lem}\cite[Lemma B2]{Hu2025_arxiv}
\label{lem_converge_samplemean}
    Let $\mu_n$ be an empirical distribution with a compact support $\mathcal{Q} \subset \mathbb{R}^m$, which converges to a distribution $\mu$. Then for any continuous and bounded function $f(\cdot)$ on $\mathcal{Q}$, it holds that $\mathbb{E}_\mu(f(x))$ exists and 
    \begin{equation*}
        \int f(x) d\mu_n(x) \to \int f(x) d\mu(x) = \mathbb{E}_\mu(f(x)),~n \to \infty,
    \end{equation*}
    where the expectation is taken over $x$ with respect to $\mu$.
    
    Further, for any continuous and bounded function $f(x,c)$ on $\mathcal{Q} \times \mathcal{Q}^0$, where $\mathcal{Q}^0$ is compact as well, it holds that $\mathbb{E}_\mu(f(x,c))$ exists and 
    \begin{equation*}
        \int f(x,c) d\mu_n(x) \to \int f(x,c) d\mu(x) = \mathbb{E}_\mu(f(x,c)),~ n \to \infty
    \end{equation*}
    uniformly on $\mathcal{Q}^0$,  where the expectation is taken over $x$ with respect to $\mu$.
\end{lem}

		\begin{lem}\cite[Equation (5-73)]{Papoulis2002}
			\label{lem_moment_cos_sin}
			Let $X$ be a Gaussian random variable with mean zero and variance $\sigma^2$. Thus, for any $k=0,1,2,...$, there holds that
			\begin{equation*}
				\mathbb{E}X^{2k} = \frac{(2k)!}{2^k(k!)}\sigma^{2k}.
			\end{equation*}
Further, there holds that
			\begin{subequations}
				\begin{align}
					&\mathbb{E}(\cos(X)) = e^{-\sigma^2/2},~\mathbb{E}(\sin(X)) = 0,~\mathbb{E}(\sin(X)\cos(X)) = 0,\\
					&\mathbb{V}(\cos(X)) = \frac{1}{2}(e^{-2\sigma^2} + 1 -2e^{-\sigma^2}),\\
					&\mathbb{V}(\sin(X)) = \frac{1}{2}(1-e^{-2\sigma^2}).\label{lm10-5}
				\end{align}
			\end{subequations}
			Moreover, the $4$-th moments of $\sin(X)$ and $\cos(X)$ are finite.
		\end{lem}
		
		\begin{proof}
			The Taylor series of $\cos(X)$ is given by
			\begin{equation*}
				\cos(X) = \sum_{k=0}^\infty\frac{(-1)^kX^{2k}}{(2k)!}
			\end{equation*}
			with convergence domain $(-\infty,\infty)$. Thus, it holds that
			\begin{align*}
				\mathbb{E}(\cos(X))&=\mathbb{E}\left[\sum_{k=0}^\infty\frac{(-1)^kX^{2k}}{(2k)!}\right]\\
				&=\sum_{k=0}^\infty\frac{(-1)^k}{(2k)!}\frac{(2k)!}{2^k(k!)}\sigma^{2k}\\
				&=e^{-\sigma^2/2}.
			\end{align*}
			
			It is obvious that $\mathbb{E}(\sin(X)) = 0$ and $\mathbb{E}(\sin(X)\cos(X)) = 0$ since $\sin(\cdot)$ and $\sin(\cdot)\cos(\cdot)$ are odd functions and $\mathbb{E}X=0$.
			
			Next we derive the variance of $\cos(X)$ and $\sin(X)$. Notice that $2X$ is a Gaussian random variable with mean 0 and variance $4\sigma^2$, it holds that
			\begin{equation*}
				\mathbb{V}(\sin(X)) = \mathbb{E}(\sin^2(X)) = \mathbb{E}\left(\frac{1-\cos(2X)}{2}\right) = \frac{1}{2}(1-e^{-2\sigma^2}).
			\end{equation*}
			Moreover, it holds that
			\begin{align*}
				&\mathbb{V}(\cos(X)) = \mathbb{E}(\cos^2(X))-(\mathbb{E}(\cos(X)))^2\\ =~& 1-\mathbb{E}(\sin^2(X)) - e^{-\sigma^2} = \frac{1}{2}(e^{-2\sigma^2} + 1 -2e^{-\sigma^2}).
			\end{align*}
			
			Notice that $\sin^4(X) = \cos(4X)/8-\cos(2X)/2+3/8$, and the mean of $\cos(2X)$ and $\cos(4X)$ exist, thus $\mathbb{E}(\sin^4(X))$ exists and is finite. The similar result holds for $\cos^4(X)$.
		\end{proof}
		
		\begin{lem}
			\cite[Lemma 2.1]{Stoica1982}
			\label{lem_eigen}
			Let both $C$ and $S$ be symmetric and positive semidefinite matrices such that $Q = C+S$ is positive definite.
			Then, $\lambda_{\max}(Q^{-1}S) = 1$ if and only if $C$ is singular.
		\end{lem}
		
		\begin{lem}\cite[Theorem 14.4-1 on page 476]{Bishop2007}
			\label{lem_sqrtn}
			Let $\{X_k\}$ be a sequence of independent random variables with $\mathbb{E}X_k = 0$ and $\mathbb{E}X_k^2 \leq C < \infty$ for all $k$ and a positive constant $C$. Then, there holds that $\frac1n\sum_{k=1}^{n}X_k = O_p(1/\sqrt{n})$.
		\end{lem}
 
	\bibliographystyle{plain}
	\bibliography{ref_abbv.bib}

\end{document}